\documentclass{article}
\usepackage{geometry}
\usepackage[utf8]{inputenc}

\usepackage[square, sort, numbers]{natbib} 

\usepackage{amsmath}
\usepackage{amssymb}
\usepackage{amsfonts}
\usepackage{amsthm}
\usepackage{epsfig, graphicx, import, color}

\usepackage{subcaption}
\usepackage[english]{babel}
\usepackage[autostyle]{csquotes}

\numberwithin{equation}{section}

\newtheorem*{prop*}{Proposition}

\usepackage{accents}
\newcommand*{\dt}[1]{%
  \accentset{\mbox{\large\bfseries .}}{#1}}

\usepackage{enumerate}

\title{Double bracket formulation for the distribution function approach to multibead-chain suspensions}
\author{Ching Lok Chong}
\date{17 September 2019 \\\quad\\OCIAM, Mathematical Institute,
University of Oxford,
Andrew Wiles Building,
\\ Radcliffe Observatory Quarter,
Woodstock Road,
Oxford,
OX2 6GG,
UK}

\begin{document}

\maketitle

\begin{abstract}
A suspension of elastic chains of small beads in a Newtonian fluid is a common model for a viscoelastic polymer solution.
The configuration of these multibead chains can be described by a distribution function that evolves according to a Liouville or Fokker--Planck equation.
For example, when we consider pairs of beads connected with Hookean springs suspended in an incompressible fluid, this approach leads to the well-known upper-convected Maxwell model. This is done by taking the second moment of the Fokker--Planck equation that governs the distribution function, which is also sufficient to describe the stress of the bead-spring pairs on the fluid.
The evolution of these multibead-chain suspensions can be described using a double bracket formulation with a Hamiltonian functional. The conservative part of the dynamics is described by a Poisson bracket, and the dissipative part by an additional symmetric bracket.
We treat the configuration space of multibead chains as a higher order tangent bundle. Lifting the fluid velocity field to the bundle leads naturally to a semidirect product Lie--Poisson bracket for the conservative dynamics.
The elastic stress exerted by the multibead chains on the fluid then follows directly from the same Hamiltonian functional that governs the internal dissipative mechanics of the multibead chains.
For chains with three or more beads, the possible bending of the chain introduces an angular momentum flux that is absent for chains with two beads. This flux appear as an asymmetric elastic stress whose antisymmetric part is the divergence of a rank-$3$ tensor, as in the Cosserats' theory of couple stresses in media with no internal angular momentum density.
We investigate the possibility of an exact closure, passing from a distribution function description to a closed internal state variable description of the fluid suspension, and obtain some sufficient conditions for their existence. The resulting exactly closable models are generalisations of the upper-convected Maxwell model to Hookean bead-spring chains instead of Hookean bead-spring pairs.
\end{abstract}

\section{Introduction}\label{sec:intro}

A common microscopic model for a non-Newtonian fluid, e.g. a viscoelastic polymer solution, is that of a \emph{fluid suspension}, which consists of small bodies with internal structure suspended in a Newtonian fluid. The evolution of these small bodies can be described by a \emph{hydrodynamic} part, where the small bodies are advected like Lagrangian markers, and a \emph{dissipative} part, which describes the internal relaxation of the small bodies, in addition to diffusive effects due to stochastic Brownian forces. The exact configurations of the small bodies are difficult to measure directly and uninteresting for macroscopic experiments, so one often resorts to a probabilistic, or ensemble, description of the small bodies, using an ensemble \emph{distribution function} $\psi(x,y)$ to encode the number density of bodies at position $x$ and with internal configuration $y$. The stress exerted on the fluid by the suspended bodies is a macroscopic quantity i.e. a function of position only, so it can often be given in terms of some statistical properties of the internal degrees of freedom in the ensemble, or in other words certain $y$-integrals (or $y$-moments) of the distribution function $\psi(x,y)$.

A familiar example of such a system would be a dilute suspension of Hookean bead-spring pairs, which serves as a microscopic model for a non-Newtonian viscoelastic fluid described by the upper-convected Maxwell model \cite{Bird87, Renardy00}. In this case, the internal configuration $y$ would be a vector describing the relative displacement of the two beads. The evolution equation for the distribution function $\psi(x,y)$ is a \emph{Liouville equation}, or a \emph{Fokker--Planck equation}, which describes the advection of the beads as Lagrangian markers, the relative motion of the beads to the background flow due to the force of the spring, and diffusion due to stochastic Brownian forces exerted by the suspending fluid. For a Hookean spring, the particle-contributed stress depends solely on the \emph{number density} $n(x) = \int \mathrm{d}^n y \ \psi(x,y)$ and the \emph{conformation tensor} $C^{jk}(x) = \int \mathrm{d}^n y \ y^jy^k\psi(x,y)$, which are the zeroth and second $y$-moments of $\psi$ respectively. Moreover, if we take these moments of the Liouville equation for $\psi$, we can show that they form a \emph{closed system} -- the time evolution of $n(x)$ and $C^{jk}(x)$ do not depend on any other $y$-moments of $\psi$. Upon closer inspection, the reduced system of equations we have obtained this way is precisely the upper-convected Maxwell model for the elastic stress.

The above observations can be recast into more abstract geometrical language. The fluid velocity field can be thought of as a vector field $\mathbf{u}$ on a manifold $M$, while the configuration space for each of the suspended bodies can be thought of as a \emph{fibre bundle} $E\rightarrow M$, which (loosely speaking) is a smoothly varying assignment of a standard fibre $F$ to each point $x$ in $M$, written as $F\mapsto F_x$. The standard fibre $F$ can be thought of as the internal configuration space of a reference copy of the small body, whereas the total space $E$ can be thought of as the full configuration space of the body in the fluid domain. Certain fibre bundles $E$ admit lifts of vector field $\mathbf{u}$ on the base manifold $M$ to a vector field $\mathbf{u}^{\#}$ on $E$, such that the lift is a \emph{Lie algebra homomorphism}. This allows us to think of the flow of the lifted vector field as the flow of Lagrangian markers in a suitable configuration space. For example, if $E=TM$ is the \emph{tangent bundle}, then the \emph{complete lift} (or \emph{tangent lift}) of vector fields describe the motion of a line element (or a tangent vector) which is frozen into a fluid -- the basepoint of the line element moves with the fluid, while the displacement across the line element is stretched by the local velocity gradient. If the evolution of the small bodies can be described purely as Lagrangian markers, then the resulting coupled dynamical system for the fluid and the (ensemble) distribution function $\psi$ can be formulated as a \emph{non-canonical Hamiltonian system}, whose dynamics is completely described by an \emph{(abstract) Poisson bracket} $\{\cdot,\cdot\}$ and a \emph{Hamiltonian functional} $H$. The system is called non-canonical, because the Poisson bracket $\{\cdot,\cdot\}$ is an abstract, coordinate-free generalisation of the usual Poisson bracket given in canonical coordinates. The Poisson brackets that we will consider belong to the family of \emph{semidirect product Lie--Poisson brackets} \cite{Marsden84a, Marsden84b}. One of the main advantages for such a formulation is that, given an arbitrary Hamiltonian functional for the suspended small bodies, the body force (and hence the stress) exerted on the fluid by the bodies can be calculated from a straightforward manipulation of the Poisson bracket \cite{Grmela89, Beris94book}. Semidirect product Lie--Poisson structures have been used to describe ideal fluids \cite{Salmon88, Morrison98}, complex fluids \cite{Grmela88, Grmela89, Beris90, Edwards90, Beris94book, Mackay19} and magnetohydrodynamics \cite{Morrison80, Marsden84a}.

The dissipative part of the dynamical system, which typically consists of the relaxation of the internal degrees of freedom e.g. motions of the body caused by internal elastic forces, as well as Brownian diffusive effects, can be modelled by a symmetric \emph{dissipation bracket} \cite{Beris90, Edwards90,Beris94book, Morrison86}. This can be thought of as an implementation of the \emph{mobility relations} in Stokes flow that convert any extra forces to a velocity difference with the background flow \cite{Beris94book, Grmela88}. Such internal interactions typically do not couple to the macroscopic fluid flow, so the procedure of calculating the stress using the Poisson bracket is still valid. This gives a complete description of the fluid suspension in terms of the usual fluid variables and the distribution function $\psi$. It is also possible to account for the Newtonian viscous stress of the fluid using dissipation brackets \cite{Morrison84,Beris94book,EnzTurski79}. For the case of a Hookean bead-spring pair suspension, the inclusion of the Newtonian viscous stress leads to the Oldroyd-B model, a generalisation of the upper-convected Maxwell model with nonzero fluid viscosity.

However, the distribution function $\psi(x,y)$ depends on both the internal and macroscopic degrees of freedom, and the evolution equation of $\psi(x,y)$ is a partial differential equation in $\mathrm{dim}(E)$ dimensions (not including time), which is much more computationally expensive to solve than a system of partial differential equations in $\mathrm{dim}(M)$ dimensions. For example, for a bead-spring pair suspension in a $3$-dimensional domain, the evolution equation for the distribution function will be a $6$-dimensional partial differential equation. Fortunately, for some cases, it is not necessary to solve for the full distribution function $\psi(x,y)$ -- it is possible that there exist a finite number of $y$-integrals (more abstractly, integrals over the fibres $F_x$) of the distribution function $\psi(x,y)$ that form a closed system of evolution equations, in addition to being able to fully describe the stress tensor. This property is called \emph{(finite and exact) closure}. When this is possible, we can understand the particle-contributed stress fully by evolving such $y$-integrals of $\psi(x,y)$, which will be $x$-dependent fields, without having to solve for the full distribution function $\psi(x,y)$. In abstract terms, the $y$-integrals of $\psi(x,y)$ will be geometric objects living on $M$ (sections of some naturally constructed fibre bundles), e.g. tensor fields. This is indeed the case for the Hookean bead-spring pair, for which there exists a closed system of evolution equations for the macroscopic number density $n(x)$ and the conformation tensor $C^{jk}(x)$. Together they are sufficient to describe the particle-contributed stress.

While the upper-convected Maxwell model and its connection with a suspension of bead-spring pairs is well known \cite{Bird87, Renardy00}, a straightforward extension of the strategy above produces models for multibead-chain suspensions, namely by looking at fibre bundles known as \emph{higher order tangent bundles} \cite{Yano73book}. Loosely speaking, a point in the $N^{th}$ order tangent bundle $T^{(N)}M$ can be thought of as a local $(N+1)$-point approximation to a path attached to $M$. This is a generalisation of the tangent bundle $TM$, which is the configuration space for tangent vectors, or material line elements. Such fibre bundles admit analogues of the complete lift of vector fields to tangent bundles, also called the \emph{complete lift}, such that the hydrodynamic part of the evolution of multibead-chains (as Lagrangian markers) can be captured by a semidirect product Lie--Poisson bracket. If we choose a dissipation bracket corresponding to a linear mobility relation to describe the response to non-hydrodynamic forces on the multibead-chains, and choose a quadratic energy function analogous to that of Hookean bead-spring pairs, we can show that the resulting system is \emph{closed} -- there exists a finite set of moments whose evolution equations do not depend on moments outside of the set, and are also sufficient to describe the stress tensor. These models can be considered as generalisations of the upper-convected Maxwell model from bead-spring pairs to bead-spring chains.

A crucial difference between a multibead-chain with $3$ or more beads and a bead-spring pair is that the former can exchange \emph{angular momentum} between fluid parcels i.e. the particle-contributed stress tensor is asymmetric. This can be understood as follows: a multibead-chain with $3$ or more beads is sensitive to the \emph{second derivative} of the background fluid velocity, which means it can detect \emph{vorticity gradients} across fluid parcels. Equivalently, the multibead-chain can \emph{bend}, in the sense that the middle beads are not necessarily aligned with the end-to-end displacement of the chain. This is an effect that is not captured by a bead-spring pair, since it is only sensitive to the local velocity gradient. We will show that for reasonable choices of the internal energy, the antisymmetric part of the stress tensor can be written as the divergence of a $3$-index tensor, which we can think of as an \emph{angular momentum flux} across material surfaces. This is consistent with the picture that the suspended multibead-chains have no inertia, and hence have zero internal angular momentum. Thus we have a family of microscopic models for an asymmetric non-Newtonian stress, which supports an angular momentum flux due to internal structure. These are microscopic realisations of the generalised continuum systems considered in \cite{CondiffDahler64, Cosserat09book,Rosensweig85} that support asymmetric stress tensors. These continuum systems allow the interactions of the internal degrees of freedom to transmit an angular momentum flux, called the \emph{couple stress}. The internal angular momentum density of the multibead-chain suspension is identically zero everywhere, so the instantaneous torque balance is maintained by the divergence of a couple stress. By contrast, in continuum models for suspensions with internal \enquote{spin} degrees of freedom such as ferrofluids, couple stresses alone cannot balance the asymmetric hydrodynamic stress, so the internal angular momentum density must evolve dynamically \cite{Shliomis74, Rosensweig85}.

The outline of the paper is as follows. In section \textbf{\ref{sec:maths}} we give an overview of the general mathematical strategy on forming semidirect product Lie--Poisson systems relevant to fluid suspensions, and in section \textbf{\ref{sec:2bead}} we will apply the strategy to recover the upper-convected Maxwell model for bead-spring pairs, which will serve as a benchmark for the double bracket systems considered later. In section \textbf{\ref{sec:TNMmodel}} we will consider higher order tangent bundles, and the resulting semidirect product Lie--Poisson bracket obtained from such a construction. As an illustration, we will work out the case for the $3$-bead chain in section \textbf{\ref{sec:3beads}}, and examine its similarities and differences with the bead-spring pair in detail. Finally, we will investigate the closure and conservation properties of the general multibead-chain model in section \textbf{\ref{sec:TNMmodelcont}}, by working out the explicit form of the Poisson bracket.

\section{Mathematical preliminaries}\label{sec:maths}
This section is a brief exposition on some of the geometrical machinery that is used in constructing Poisson brackets relevant for the dynamics of fluid suspensions, so that we can formulate the conservative part of the dynamics as a \emph{non-canonical Hamiltonain system}. The main references are \cite{Marsden83, Marsden84a, Marsden84b, Khesin08, Marsden13book, Morrison98}. Most technical hypotheses will be suppressed. In particular, various smoothness and decay assumptions on infinite-dimensional spaces will be implicit. We will assume our manifolds to be orientable to simplify the theory of integration. To extend our exposition to non-orientable manifolds, replace every instance of \enquote{top-degree differential form} or \enquote{volume form} with \enquote{density}.

A complementary approach, which we will not pursue in this exposition, formulates the conservative part of the suspension dynamics in terms of \emph{Hamilton's principle of extremal action}. The review article \cite{Salmon88} contains an exposition on variational principles for ideal fluids. These are usually formulated in terms of Lagrangian particle labels and their material time derivatives. 
The \emph{Euler--Poincar\'e variational principle} is a reduced description using only Eulerian fields \cite{Holm98,Marsden13book}. It only considers variations of the Eulerian fields that can be generated by variations of the Lagrangian label-to-particle map. This is the variational analogue of non-canonical Hamiltonian mechanics formulated using a \emph{Lie--Poisson bracket} (see section \textbf{\ref{subsec:liepoi}}).
The {Euler--Poincar\'e variational principle} has been used to study the geometric structure of various types of ideal complex fluids in \cite{Holm02, GayBalmaz09}.

\subsection{Lie algebras}\label{subsec:liealg}
An \emph{abstract Lie algebra} $\mathfrak{g}$ is a vector space together with a bilinear $\mathfrak{g}$-valued map $[\cdot ,\cdot ]: \mathfrak{g} \times \mathfrak{g} \rightarrow \mathfrak{g}$ called the \emph{Lie bracket}, such that, for all $u,v,w \in \mathfrak{g}$:
\begin{align}
[u,v] &= -[v,u],  & (\emph{antisymmetry})\\
[[u,v],w] + [[v,w],u] + [[w,u],v] &= 0. & (\emph{Jacobi identity}) 
\end{align}
Our primary example will be the \emph{Lie algebra of vector fields} on a manifold $M$, denoted $\mathrm{Vect}(M)$.
In tensor calculus and differential geometry, one associates vector fields $\mathbf{u}$, whose components will be written as $u^{i}$, with first order differential operators $u^{i}{\partial}/{\partial x^i}$ corresponding to the \emph{directional derivative} of a function on $M$ along $\mathbf{u}$. This is sometimes written as $\mathbf{u} \cdot \nabla$.

The Lie bracket in this case is the commutator of first order differential operators:
\begin{align}
    [u,v]^i\frac{\partial f}{\partial x^i} :&= u^j\frac{\partial}{\partial x^j}\left(v^i\frac{\partial f}{\partial x^i}\right) - v^j\frac{\partial}{\partial x^j}\left(u^i\frac{\partial f}{\partial x^i}\right) = \left(u^j\frac{\partial v^i}{\partial x^j} - v^j\frac{\partial u^i}{\partial x^j}\right)\frac{\partial f}{\partial x^i}.
\end{align}
The components of Lie bracket of two vector fields $[u,v]^i$ transforms like a contravariant vector (i.e. like $u^i$) under a change of coordinates, so it is a vector field. The Jacobi identity can be also checked by direct verification. This shows that $\mathrm{Vect}(M)$ endowed with the commutator bracket is indeed a Lie algebra.

The above identities for an abstract Lie algebra originate from the commutator of a Lie group, or in looser language the commutator of infinitesimal continuous symmetry transformations.
The corresponding infinitesimal symmetry transformations for a given vector field $u^i$ is the \emph{flow} of $u^i$, i.e. transformations on the manifold $M$ constructed by following the flow lines of the vector field $u^i$ (up to some technical conditions).

Other (finite-dimensional) examples of Lie algebras are $\mathfrak{gl}(n)$, the Lie algebra of $n \times n$ real (or complex) matrices, or $\mathfrak{so}(n)$, the Lie algebra of $n \times n$ antisymmetric matrices. In both cases, the Lie bracket is the matrix commutator $[A,B] = AB-BA$. In particular, $\mathfrak{so}(3)$ is also isomorphic to the cross product algebra on $\mathbb{R}^3$.

\subsection{Non-canonical Hamiltonian mechanics}\label{subsec:hammec}
Many conservative physical systems can be formulated in terms of non-canonical Hamiltonian systems. We will provide a brief sketch of the key ideas, without delving too deeply into the mathematical details. In particular, we shall only concern ourselves with the formal aspects of the theory of functionals. For a more complete exposition of non-canonical Hamiltonian mechanics in the context of infinite-dimensional Lie groups from a mathematical point of view, see \cite{Khesin08,Marsden83,Marsden84b}. The review articles \cite{Salmon88, Morrison98, ZakharovKuznetsov97} cover non-canonical Hamiltonian mechanics from a physical point of view.

\subsubsection{Functional derivatives}\label{subsubsec:fundif}
Let $V$ be a Banach or Frech\'{e}t space, not necessarily finite dimensional. Let $V^{*}$ be its (smooth) dual. We will implicitly assume smoothness and decay properties at infinity as we require throughout. If $F: V \rightarrow \mathbb{R}$ is a functional on $V$, then its \emph{functional derivative} at $v\in V$ is an element ${\delta F}/{\delta v} \in V^{*}$, such that
\begin{align}
    \left\langle \frac{\delta F}{\delta v}[v] , w \right\rangle = \lim_{\epsilon \rightarrow 0} \frac{F[v + \epsilon w] - F[v]}{\epsilon},
\end{align}
where $\langle \cdot,\cdot\rangle$ is the dual pairing between $V$ and $V^{*}$. 

Note that if $G[\alpha]$ is a functional on the dual space $V^{*}$, where $\alpha\in V^{*}$, then ${\delta G}/{\delta \alpha}$ will be an element in the double dual $V^{**}$, which in general is not naturally isomorphic to $V$. Nonetheless, there is a natural embedding $V \hookrightarrow V^{**}$, and henceforth we shall assume $F$ is sufficiently regular, so that ${\delta G}/{\delta \alpha}$ lies in the image of the embedding of $V$ for all $\alpha$ in a smooth subspace of $V^{*}$ with sufficient decay properties.

For example, if $C^\infty(\mathbb{R}^n)$ is the space of smooth functions on $\mathbb{R}^n$, then its (smooth) dual space can be thought of as the space of smooth $n$-forms (or oriented volume elements) $\Omega^n(\mathbb{R}^n)$, with typical element $\alpha \mathrm{d}^n {x}$, and the pairing between $f \in C^\infty (\mathbb{R}^n)$ and $\alpha \mathrm{d}^n {x} \in \Omega^n(\mathbb{R}^n)$ is by \emph{integration}:
\begin{align}
    \left\langle \alpha \mathrm{d}^n {x}, f \right\rangle = \int \mathrm{d}^n {x} \ \alpha f .
\end{align}
If $F[f] = \int \mathrm{d}^n {x} \  \varphi \left( f, {\partial f}/{\partial x^i}\right)$, then integrating by parts and assuming boundary terms vanish gives
\begin{align}
     \frac{\delta F}{\delta f} = \left(\frac{\partial \varphi}{\partial f} - \frac{\partial}{\partial x^i} \frac{\partial \varphi}{\partial \left(\frac{\partial f}{\partial x^i}\right)}\right) \mathrm{d}^n {x} ,
\end{align}
which is the usual Euler--Lagrange variational derivative of functionals. The factor $\mathrm{d}^n {x}$ is usually suppressed -- this reflects a choice of a standard volume element $\mathrm{d}^n {x}$ in $\mathbb{R}^n$.

Similarly, if $F[\alpha \mathrm{d}^n {x}] = \int \mathrm{d}^n {x} \  \varphi \left( \alpha, {\partial \alpha}/{\partial x^i}\right)$, then integrating by parts and assuming boundary terms vanish gives
\begin{align}
     \frac{\delta F}{\delta \alpha} = \frac{\partial \varphi}{\partial \alpha} - \frac{\partial}{\partial x^i} \frac{\partial \varphi}{\partial \left(\frac{\partial \alpha}{\partial x^i}\right)} ,
\end{align}
this time without the volume factor $\mathrm{d}^n{x}$, since ${\delta F}/{\delta \alpha}$ lives in the space of functions, which is dual to the space of volume forms.

Although we will later suppress the volume factor $\mathrm{d}^n{x}$ to avoid notational clutter, it is useful to keep it in mind, since it can keep track of factors of the Jacobian determinant when we change coordinates.

\subsubsection{Poisson brackets}\label{subsubsec:poibra}
Consider the space of functionals $\mathcal{F}(V)$ on a vector space $V$. An (abstract) \emph{Poisson bracket} is a bilinear operation on functionals, $\{\cdot , \cdot \}: \mathcal{F}(V) \times \mathcal{F}(V) \rightarrow \mathcal{F}(V)$, such that for all $F,G,H \in \mathcal{F}(V)$:
\begin{align}
    \{F,G\} &= -\{G,F\}, &(\emph{antisymmetry}) \\
    \{F,GH\} &= \{F,G\}H + \{F,H\}G, &(\emph{Leibniz/product rule}) \\
    0 &= \{\{F,G\},H\} + \{\{G,H\},F\} + \{\{H,F\},G\}, &(\emph{Jacobi identity})
\end{align}
whenever the expressions are well-defined. The resemblance with Lie algebra is more than superficial -- there is a crucial Poisson bracket that can be constructed from a Lie algebra, called the \emph{Lie--Poisson bracket}, which we will make use of extensively.

In usual (finite-dimensional) Hamiltonian mechanics on $\mathbb{R}^{2n}$ with canonical coordinates $(\mathbf{q},\mathbf{p})$, we have the \emph{canonical Poisson bracket} $\{\cdot,\cdot\}$, defined by
\begin{align}
    \{F,G\} = \frac{\partial F}{\partial \mathbf{q}} \cdot \frac{\partial G}{\partial \mathbf{p}} - \frac{\partial G}{\partial \mathbf{q}} \cdot \frac{\partial F}{\partial \mathbf{p}}.
\end{align}
Functions $F(\mathbf{q},\mathbf{p})$ are evolved in time with a \emph{Hamiltonain function} $H = H(\mathbf{q},\mathbf{p})$, given in terms of the canonical Poisson bracket as
\begin{align}
    \dt{F} = \{F,H\}.
\end{align}
By using the chain rule $\dt{F} = ({\partial F}/{\partial \mathbf{q}}) \cdot \dt{\mathbf{q}} + ({\partial F}/{\partial \mathbf{p}}) \cdot \dt{\mathbf{p}}$, and letting $F$ be arbitrary, we obtain the trajectory of a point in phase space as
\begin{align}
    \dt{\mathbf{q}} &= \{\mathbf{q},H\} =  \frac{\partial H}{\partial \mathbf{p}}, \\
    \dt{\mathbf{p}} &= \{\mathbf{p},H\} =  -\frac{\partial H}{\partial \mathbf{q}}.
\end{align}
This argument using the chain rule to obtain the phase space trajectory can be generalised to infinite-dimensional systems in a straightforward manner, as we shall see below.
The Poisson bracket can be thought of as a structure that encodes the fact that the coordinates are canonically conjugate, while the Hamiltonian function contains the energetic, or dynamic, information of the system.

In light of this, we can define a \emph{non-canonical Hamiltonian system} on a space of sufficiently regular functionals $\mathcal{F}(V)$ as follows:

The time evolution of functionals $F[v]$ is given in terms of an (abstract) Poisson bracket $\{\cdot , \cdot \}$ and a chosen Hamiltonian functional $H[v]$ as follows:
\begin{align}
    \dt{F}[v] = \{F,H\}[v].
\end{align}
We can obtain the trajectory of a point $w \in V$ using the chain rule for arbitrary $F[w]$:
\begin{align}
\label{chainrule}
\left\langle \frac{\delta F}{\delta v}[w] , \dt{w} \right\rangle = \dt{F}[w] =  \{F,H\}[w] = \left\langle \frac{\delta F}{\delta v}[w] , J[w]\left(\frac{\delta H}{\delta v}[w]\right) \right\rangle.
\end{align}
The bracket is guaranteed to be bilinear on ${\delta F}/{\delta v}, {\delta H}/{\delta v}$, by the axioms of an abstract Poisson bracket, so we can always write $\{F,H\}$ in the form shown in the last expression in (\ref{chainrule}), using the linear operator $J[w]:V^{*}\rightarrow V$ called the \emph{Poisson tensor}. In contrast to Hamiltonian mechanics on $\mathbb{R}^{2n}$ in canonical coordinates, the Poisson tensor can vary with $w$ in general. Thus we have
\begin{align}
    \dt{w} = J[w]\left( \frac{\delta H}{\delta v}[w]\right).
\end{align}
The reason for calling the dynamics \emph{non-canonical} is that the Poisson bracket need not be expressible, even locally, as a canonical Poisson bracket. In fact, this formalism allows for functionals $C[v]$ that are not locally constant such that $\{F,C\}=0$ for arbitrary functionals $F$. Any such $C$ is called a \emph{Casimir functional}. The local theory of non-canonical Poisson brackets in finite dimensions has been extensively studied in \cite{Weinstein83}.

Casimir functionals can be thought of as kinematic invariants of the Poisson bracket, independent of the energy of the system. For example, the total mass of a compressible ideal fluid is one such invariant. When Casimir functionals (which are not locally constant) exist, the Poisson bracket is called \emph{degenerate}. It is the existence of such kinematically conserved quantities that distinguishes canonical and non-canonical Hamiltonian mechanics.

\subsection{Lie--Poisson dynamics on the dual of a Lie algebra}\label{subsec:liepoi}
The Jacobi identity for Poisson brackets on functionals is typically difficult and unenlightening to verify from direct computation, and as such they are difficult to construct {ab initio}. In the following we will construct the \emph{Lie--Poisson bracket} associated to a Lie algebra $\mathfrak{g}$, which is a non-canonical Poisson bracket for functionals on the dual $\mathfrak{g}^*$ of $\mathfrak{g}$. This provides a way to construct a wide class of Poisson brackets that automatically satisfy the Jacobi identity.

Given a Lie algebra $\mathfrak{g}$, consider the (smooth) dual $\mathfrak{g}^{*}$ of $\mathfrak{g}$, with duality pairing written as $\langle\cdot , \cdot \rangle$ as usual.

For example, if $\mathfrak{g} = \mathrm{Vect}(M)$ is the Lie algebra of vector fields $u^i$ on a manifold $M$, the dual space will consist of elements of the form $m_i \mathrm{d}^n {x}$, where the dual pairing is given by integration:
\begin{align}
    \left\langle m_i \mathrm{d}^n {x}, u^i \right\rangle = \int \mathrm{d}^n{x} \ m_i u^i .
\end{align}
The elements $m_i \mathrm{d}^n {x}$ can be thought of as the (canonical) momentum density of a fluid when $u^i$ is the vector field generating the flow of a fluid. The factor $\mathrm{d}^n {x}$ is usually suppressed from the notation.

In this set-up, the space of sufficiently regular functionals $\mathcal{F}(\mathfrak{g}^{*})$ on the dual $\mathfrak{g}^{*}$ of the Lie algebra $\mathfrak{g}$ is automatically equipped with a Poisson bracket, called the \emph{$(\pm)$-Lie--Poisson bracket}:
\begin{align}
\{F,G\}_{LP, \pm}[\mathbf{m}] := \pm \left\langle \mathbf{m}, \left[ \frac{\delta F}{\delta \mathbf{m}}[\mathbf{m}], \frac{\delta G}{\delta \mathbf{m}}[\mathbf{m}] \right] \right\rangle ,
\end{align}
where $[\cdot , \cdot ]$ is the Lie bracket on $\mathfrak{g}$. Note that in general ${\delta F}/{\delta \mathbf{m}}$ is an element in $\mathfrak{g}^{**}$, but we will only work with sufficiently regular $F$, such that ${\delta F}/{\delta \mathbf{m}}$ lies in the image of the natural embedding $\mathfrak{g}\hookrightarrow \mathfrak{g}^{**}$. When this condition is met, the Lie brackets of such functional derivatives are well-defined.

The fact that $[\cdot , \cdot ]$ satisfies the axioms of a Lie bracket implies that $\{\cdot,\cdot \}_{LP, \pm}$ satisfies the axioms of a Poisson bracket, which can be proved by direct computation, or by reduction from a canonical Poisson bracket. The details of each argument can be found in \cite{Marsden83, Marsden84b, Marsden13book}, but we will provide a brief sketch here:

Given a Lie group $G$, we can form the \emph{cotangent bundle} $T^{*}G$ of $G$. All cotangent bundles are naturally equipped with a \emph{canonical symplectic $2$-form}, analogous to the $2$-form $\omega = \mathrm{d}\mathbf{q} \wedge \mathrm{d}\mathbf{p}$ in usual classical mechanics. Since $G$ is a Lie group, it acts on itself canonically via left (or right) multiplication, and this action can be lifted to the cotangent bundle $T^{*}G$ of $G$. By left (or right) translating all the fibres of $T^{*}G$ to the identity element $e\in G$, and using the theory of \emph{momentum maps} for Poisson manifolds (sometimes called \emph{Poisson maps} in this context to distinguish between the case for symplectic manifolds), we obtain a Poisson bracket on $T^{*}_eG \simeq \mathfrak{g}^{*}$, which is precisely the Lie--Poisson bracket. Furthermore, if the Hamiltonian $H$ is left-invariant (or right-invariant), then the aforementioned translation process can be applied to $H$, so that it becomes a function on $\mathfrak{g}^{*}$ only.

When applied to $G=SO(3)$, this construction reproduces the fact that the dynamics of a free rigid body can be completely described by the three components of its angular momentum vector (together with  a set of initial conditions for the $3$ Euler angles).

\subsection{Semidirect product Lie algebras}\label{subsec:semdir}
Knowing that a Lie algebra $\mathfrak{g}$ automatically induces a Lie--Poisson bracket for functionals on its dual $\mathcal{F}(\mathfrak{g}^*)$, the next step is to find a strategy to enlarge the Lie algebra $\mathfrak{g}$, so that we can obtain new but related Lie--Poisson dynamical systems from old ones.

The two main strategies used in physical applications are forming \emph{semidirect products} and making \emph{central extensions}. We will focus on the first one, since we will later use the semidirect product construction to construct Lie--Poisson dynamical systems that describe some dilute suspensions in fluids with internal degrees of freedom. The Poisson bracket for ideal compressible hydrodynamics is also constructed this way.

Although we will not pursue central extensions in this exposition, we should mention that they are used in constructing the \emph{Virasoro algebra}, which is relevant to the Hamiltonian structure of the Korteweg--de Vries equation. For a thorough account, see \cite{Khesin08}.

Let $\mathfrak{g}$ be a Lie algebra, $V$ a vector space, and let $\mathrm{End}(V)$ denote the space of linear maps on $V$. A \emph{representation} of $\mathfrak{g}$ on $V$ is an assignment $\phi : \mathfrak{g} \rightarrow \mathrm{End}(V)$, such that, for all $X,Y \in \mathfrak{g}$ and all $v \in V$,
\begin{align}
    \phi(X+Y) \cdot v &= \phi(X) \cdot v + \phi(Y) \cdot v, &(\emph{linearity}) \\
    \phi([X,Y]) \cdot v &= \left[\phi(X) , \phi(Y) \right]\cdot v, &(\emph{homomorphism property})
\end{align}
where $\left[\phi(X) , \phi(Y) \right] = \phi(X) \phi(Y) - \phi(Y) \phi(X)$ is the operator commutator of linear maps on $V$.
The main example is the action of vector fields on smooth functions, given by 
\begin{align}
    \phi(\mathbf{u}) \cdot f = u^i \frac{\partial f}{\partial x^i}, \quad \text{for $\mathbf{u}\in\mathrm{Vect}(M)$ and $f\in C^\infty(M)$} .
\end{align}
A direct computation shows that $\phi$ is indeed a representation. The symbol $\phi$ is sometimes suppressed, when the context is clear. So $\mathbf{u} \cdot f$ will mean the vector field $\mathbf{u}$ acting on the function $f$ as a differential operator, which is more commonly written as $\mathbf{u} \cdot \nabla f$. We will avoid the latter notation because we will need vector fields to act on other objects in less obvious ways later.

Given a Lie algebra $\mathfrak{g}$ and a representation $\phi$ of $\mathfrak{g}$ on a vector space $V$, we can form the \emph{semidirect product Lie algebra} $\mathfrak{g}_s = \mathfrak{g} \ltimes V$. As a vector space, $\mathfrak{g} \ltimes V$ can be considered as the direct sum $\mathfrak{g} \oplus V$, so we can write a general element of $\mathfrak{g} \ltimes V$ as $(X,v)$ for $X \in \mathfrak{g}, v \in V$.

The Lie bracket on $\mathfrak{g} \ltimes V$ is defined by 
\begin{align}
    [(X,v),(Y,w)] := ([X,Y], \phi(X) \cdot w - \phi(Y) \cdot v ), \quad \text{for $(X,v),(Y,w) \in \mathfrak{g} \ltimes V$.}
\end{align}
It can be checked that this is indeed a Lie bracket -- it is manifestly bilinear and antisymmetric, and the Jacobi identity can be deduced from the homomorphism property of $\phi$. 

For example, in the semidirect product Lie algebra $\mathrm{Vect}(M) \ltimes C^\infty(M)$, the Lie bracket is given by
\begin{align}
    [(u^i,f),(v^i,g)] = ([u,v]^i, u^i\frac{\partial g}{\partial x^i} - v^i\frac{\partial f}{\partial x^i}),
\end{align}
for $\mathbf{u},\mathbf{v} \in \mathrm{Vect}(M)$ and $f,g \in C^\infty(M)$.

If we consider a Lie algebra $\mathfrak{g}$ and multiple representations $\phi_a: \mathfrak{g} \rightarrow \mathrm{End}(V_a)$ of $\mathfrak{g}$ on $V_a$ ($a = 1, \ldots N$), then the semidirect product Lie algebra with all such representations $\mathfrak{g} \ltimes (V_1 \oplus \ldots V_N)$ can be thought to be $\mathfrak{g} \oplus V_1 \oplus \ldots V_N$ as a vector space, with the following Lie bracket:
\begin{align}
    [(X,v_1,\ldots,v_N),(Y,w_1,\ldots,w_N)] =([X,Y], \phi_1(X) \cdot w_1 - \phi_1(Y) \cdot v_1, \ldots, \phi_N(X) \cdot w_N - \phi_N(Y) \cdot v_N) ,
\end{align}
for $X,Y\in \mathfrak{g}$ and $v_a,w_a \in V_a$.

When $\mathfrak{g} = \mathrm{Vect}(M)$, this physically corresponds to attaching more \enquote{advected degrees of freedom} to the fluid system when we consider the Lie--Poisson dynamics on the dual of the semidirect product Lie algebra. This has been used to construct, for example, the ideal compressible magnetohydrodynamics (MHD) equations, in which the magnetic field is \enquote{frozen} into the fluid \cite{Morrison80}.

\subsection{The Lie--Poisson bracket on the dual of a semidirect product Lie algebra}\label{subsec:liepoisem}
Let $\mathfrak{g}_s = \mathfrak{g} \ltimes (V_1 \oplus \ldots \oplus V_N)$ be a semidirect product Lie algebra, with notation as above. The dual space $\mathfrak{g}_s^*$ of $\mathfrak{g}_s$ is isomorphic to $\mathfrak{g}_s^* = \mathfrak{g}^* \oplus V_1^* \oplus \ldots \oplus V_N^*$ as a vector space. The Lie--Poisson bracket for functionals $\mathcal{F}(\mathfrak{g}_s^*)$ can be written as, for $\mathbf{m} \in \mathfrak{g}^*$ and $\mu_a \in V_a^*$ for $a = 1,\ldots,N$:
\begin{align}
\{F,G\}_{LP,\pm}[\mathbf{m},\mu_1,\ldots\,\mu_N] 
=& \pm \left\langle (\mathbf{m},\mu_1,\ldots\,\mu_N), \left[ \left(\frac{\delta F}{\delta \mathbf{m}},\frac{\delta F}{\delta \mu_1},\ldots,\frac{\delta F}{\delta \mu_N}\right), \left(\frac{\delta G}{\delta \mathbf{m}},\frac{\delta G}{\delta \mu_1},\ldots,\frac{\delta G}{\delta \mu_N}\right)\right] \right\rangle, \nonumber \\
=& \pm \left( \left\langle \mathbf{m},\left[\frac{\delta F}{\delta \mathbf{m}},\frac{\delta G}{\delta \mathbf{m}}\right] \right\rangle + \sum_{a = 1}^{N} \left\langle \mu_a, \phi_a\left(\frac{\delta F}{\delta \mathbf{m}}\right) \cdot \frac{\delta G}{\delta \mu_a} - \phi_a\left(\frac{\delta G}{\delta \mathbf{m}}\right) \cdot \frac{\delta F}{\delta \mu_a} \right\rangle \right).
\end{align}
We immediately notice several properties of this bracket:
\begin{enumerate}
\item When we add a representation $V_a$ of $\mathfrak{g}$ to the semidirect product, the resulting Lie--Poisson bracket on functionals $\mathcal{F}(\mathfrak{g}_s^*)$ acquires a new piece
\begin{align}
\pm \left\langle \mu_a, \phi_a\left(\frac{\delta F}{\delta \mathbf{m}}\right) \cdot \frac{\delta G}{\delta \mu_a} - \phi_a\left(\frac{\delta G}{\delta \mathbf{m}}\right) \cdot \frac{\delta F}{\delta \mu_a} \right\rangle = \{F,G\}_{\mu_a}.
\end{align}
This piece is called the \emph{sub-bracket due to $V_a$} (or by abuse of notation $\mu_a$). While it is not a Poisson bracket in itself, such pieces can be added to Lie--Poisson brackets constructed from semidirect products of $\mathfrak{g}$ without violating the Jacobi identity, corresponding to adding more representations to the semidirect product.
\item Suppose one has a Hamiltonian functional $H=H[\mathbf{m},\mu_1,\ldots,\mu_N]$. The noncanonical Lie--Poisson dynamics on the space of functionals $\mathcal{F}(\mathfrak{g}_s^*)$ given by this Hamiltonian is $\dt{F} = \{F,H\}$. In particular:
\begin{align}
&\left\langle \frac{\delta F}{\delta \mathbf{m}}, \dt{\mathbf{m}} \right\rangle + \sum_{a=1}^{N} \left\langle \frac{\delta F}{\delta \mu_a}, \dt{\mu}_a \right\rangle \nonumber \\
=& \pm \left( \left\langle \mathbf{m},\left[\frac{\delta F}{\delta \mathbf{m}},\frac{\delta H}{\delta \mathbf{m}}\right] \right\rangle + \sum_{a = 1}^{N} \left\langle \mu_a, \phi_a\left(\frac{\delta F}{\delta \mathbf{m}}\right) \cdot \frac{\delta H}{\delta \mu_a} \right\rangle - \sum_{a = 1}^{N} \left\langle \mu_a, \phi_a\left(\frac{\delta H}{\delta \mathbf{m}}\right) \cdot \frac{\delta F}{\delta \mu_a} \right\rangle\right).
\end{align}
Since $F$ is arbitrary, we can take $F=F[\mu_b]$ only for some $b = 1,\ldots N$. This gives
\begin{align}
\left\langle \frac{\delta F}{\delta \mu_b}, \dt{\mu_b} \right\rangle = \mp \left\langle \mu_b, \phi_b\left(\frac{\delta H}{\delta \mathbf{m}}\right) \cdot \frac{\delta F}{\delta \mu_b} \right\rangle = \mp \left\langle \phi_b\left(\frac{\delta H}{\delta \mathbf{m}}\right)^\dagger\cdot\mu_b,  \frac{\delta F}{\delta \mu_b} \right\rangle,
\end{align}
where $\phi_b\left({\delta H}/{\delta \mathbf{m}}\right)^\dagger : V_b^* \rightarrow V_b^*$ is the (formal) adjoint to $\phi_b\left({\delta H}/{\delta \mathbf{m}}\right): V_b \rightarrow V_b$. For the applications relevant to fluid dynamics, \enquote{taking the adjoint} usually means integrating by parts, assuming boundary terms vanish. This gives the dynamics as
\begin{align}\label{advect}
\dt{\mu_b} = \mp \phi_b\left(\frac{\delta H}{\delta \mathbf{m}}\right)^\dagger\cdot\mu_b.
\end{align}
In usual fluid dynamics, ${\delta H}/{\delta \mathbf{m}} = \mathbf{u}$ will be the fluid velocity. So we can think of (\ref{advect}) as corresponding to $\mu_b$ being \enquote{advected} by ${\delta H}/{\delta \mathbf{m}}$, as Lagrangian markers attached to fluid parcels.
\end{enumerate}

As an example, consider the Lie--Poisson dynamics for functionals on $(\mathrm{Vect}(\mathbb{R}^n) \ltimes C^\infty(\mathbb{R}^n))^*$. A typical element of $(\mathrm{Vect}(\mathbb{R}^n) \ltimes C^\infty(\mathbb{R}^n))^*$ can be written as a pair $(m_i,\rho)$ (the volume element $\mathrm{d}^n{x}$ has been suppressed), considered as the fluid momentum density and fluid mass density, respectively.

The minus Lie--Poisson bracket is, for functionals $F,G \in \mathcal{F}((\mathrm{Vect}(\mathbb{R}^n) \ltimes C^\infty(\mathbb{R}^n))^*)$:
\begin{align}
   & \{F,G\}[\mathbf{m},\rho] = - \left\langle \mathbf{m},\left[\frac{\delta F}{\delta \mathbf{m}},\frac{\delta G}{\delta \mathbf{m}}\right] \right\rangle - \left\langle \rho, \frac{\delta F}{\delta \mathbf{m}} \cdot \frac{\delta G}{\delta \rho} - \frac{\delta G}{\delta \mathbf{m}} \cdot \frac{\delta F}{\delta \rho} \right\rangle, \nonumber \\
    =& - \int \mathrm{d}^n{x} \ m_i \left[\frac{\delta F}{\delta m_j}\frac{\partial}{\partial x^j}\left(\frac{\delta G}{\delta m_i}\right) - \frac{\delta G}{\delta m_j}\frac{\partial}{\partial x^j}\left(\frac{\delta F}{\delta m_i}\right) \right] - \int \mathrm{d}^n{x} \ \rho \left[\frac{\delta F}{\delta m_j}\frac{\partial}{\partial x^j}\left(\frac{\delta G}{\delta \rho}\right) - \frac{\delta G}{\delta m_j}\frac{\partial}{\partial x^j}\left(\frac{\delta F}{\delta \rho}\right) \right] .
\end{align}
If we choose the Hamiltonian functional
\begin{align}\label{Hfluids}
    H = H_{fluids}[\mathbf{m},\rho] = \int \mathrm{d}^n{x} \ \left(\frac{m_i m_j \delta^{ij}}{2\rho} + \rho U(\rho)\right)
\end{align}
and write $u^i = {\delta H}/{\delta m_i} = {m_j}\delta^{ij}/{\rho}$, then after some computation, one recovers the usual equations for ideal compressible hydrodynamics. (See, for example, \cite{Marsden84b, Morrison80}.)

\section{The bead-spring pair equations}\label{sec:2bead}
This section covers the formulation of the Lie--Poisson bracket for the distribution function of bead-spring pairs. It turns out that the configuration space of bead-spring pairs is most naturally described as a tangent bundle, and the Lie--Poisson bracket obtained coincides with the one in \cite{Grmela88}. The double bracket formulation is also considered in order to incorporate dissipative and diffusive effects, and it turns out that the kinetic equation for bead-spring pairs in \cite{Bird87, Renardy00} can be recovered this way.
\subsection{Tangent bundles}\label{subsec:defineTM}
The configuration space of small bead-spring pairs is most naturally described as a tangent bundle. We will briefly recall the notion of a tangent bundle and explain why this is the case.

Consider an $n$-dimensional manifold $M$ (it is sufficient to take $M=\mathbb{R}^n$ for the applications later), and the family of \emph{smooth paths} on $M$, i.e. smooth maps $\alpha : I \rightarrow M$, where $I$ is a (fixed) closed interval in $\mathbb{R}$ containing $0$ in its interior.

The space of all smooth paths on $M$ is generally infinite-dimensional, which makes the space of functions on the space of paths difficult to deal with. Thus we look for finite-dimensional approximations to such a space.

Consider the equivalence class of paths under the equivalence relation $\stackrel{(1)}{\sim}$, given by
\begin{align}
    \alpha \stackrel{(1)}{\sim} \beta \ \text{if and only if} \ \alpha (0) = \beta (0), \ \text{and} \ \alpha{'}(0) = \beta{'}(0) .
\end{align}
If $\alpha^{i}$ and $\beta^{i}$ are the coordinates for the paths $\alpha, \beta$ respectively, then the equivalence relation is given by $2n$ equations:
\begin{align}
    \alpha^{i}(0) &= \beta^{i}(0), \\
    \frac{\mathrm{d}\alpha^i }{\mathrm{d}t}(0) &= \frac{\mathrm{d}\beta^i}{\mathrm{d}t}(0).
\end{align}
One can check that this equivalence relation is coordinate independent. We will denote each equivalence class by $(x^i,y^i)$, where $x^i$ is the coordinates of $\alpha(0)$ and $y^i$ the coordinates of $\mathrm{d}\alpha/\mathrm{d}t(0)$, for $\alpha(t)$ a representing element of the equivalence class. Moreover, by constructing a path $\alpha^i(t) = x^i + y^it + O(t^2)$ using coordinates, we can show that each of these equivalence classes is nonempty. The space of all such equivalence classes is called the \emph{tangent bundle} $TM$ of $M$, and is a $2n$-dimensional manifold.

A general point on the tangent bundle $TM$ can be thought of as a point $x^i$ on $M$, together with a tangent vector $y^i$ attached to the point $x^i \in M$. Informally, this is like a two-point approximation to a small segment of a curve in $M$.
This physical picture of a tangent bundle suggests that $TM$ is indeed the correct configuration space of small bead-spring pairs living on $M$.

\subsection{Complete/Tangent lift of vector fields on $M$ to $TM$}\label{subsec:liftTM}
Since $TM$ is constructed naturally from $M$, we expect vector fields $\mathbf{u} \in \mathrm{Vect}(M)$ on $M$ to act on geometrical objects living on $TM$ (functions, tensors, etc.) in a nice way. This will in fact allow us to construct a semidirect product $\mathrm{Vect}(M) \ltimes C^\infty(TM)$ of vector fields on $M$ acting on functions on $TM$.

The dual space $(C^\infty(TM))^*$ can be thought of as $\Omega^{2n}(TM)$, the space of volume forms on $TM$ with typical element $\psi(x,y) \mathrm{d}^n{x}\mathrm{d}^n{y}$, with integration being the dual pairing. It has a natural interpretation as a \emph{distribution function} of bead-spring pairs in configuration space. As we shall see later, the semidirect product formulation will describe the dynamics of the bead-spring pairs as Lagrangian markers embedded in the fluid.

The vector field $\mathbf{u}\in\mathrm{Vect}(M)$ induces a vector field $\mathbf{u}^{\#}\in\mathrm{Vect}(TM)$, called the \emph{complete lift} or \emph{tangent lift} \cite{Yano73book}, which can be described as follows. Consider how the flow of $u^i$ affects a path $\alpha^i(t)$ on $M$. For some small flow parameter $s$, the path $\alpha^i(t)$ will be deformed by the flow, as in figure \ref{fig:completelift}):
\begin{align}
\widetilde{\alpha}^i(t,s) = \alpha^i(t) + su^i(\alpha(t)) + O(s^2)
\end{align}

\begin{figure}
\centering
\begingroup%
  \makeatletter%
  \providecommand\color[2][]{%
    \errmessage{(Inkscape) Color is used for the text in Inkscape, but the package 'color.sty' is not loaded}%
    \renewcommand\color[2][]{}%
  }%
  \providecommand\transparent[1]{%
    \errmessage{(Inkscape) Transparency is used (non-zero) for the text in Inkscape, but the package 'transparent.sty' is not loaded}%
    \renewcommand\transparent[1]{}%
  }%
  \providecommand\rotatebox[2]{#2}%
  \newcommand*\fsize{\dimexpr\f@size pt\relax}%
  \newcommand*\lineheight[1]{\fontsize{\fsize}{#1\fsize}\selectfont}%
  \ifx\svgwidth\undefined%
    \setlength{\unitlength}{300.49190324bp}%
    \ifx\svgscale\undefined%
      \relax%
    \else%
      \setlength{\unitlength}{\unitlength * \real{\svgscale}}%
    \fi%
  \else%
    \setlength{\unitlength}{\svgwidth}%
  \fi%
  \global\let\svgwidth\undefined%
  \global\let\svgscale\undefined%
  \makeatother%
  \begin{picture}(1,0.4)%
    \lineheight{1}%
    \setlength\tabcolsep{0pt}%
    \put(0,0){\includegraphics[width=\unitlength]{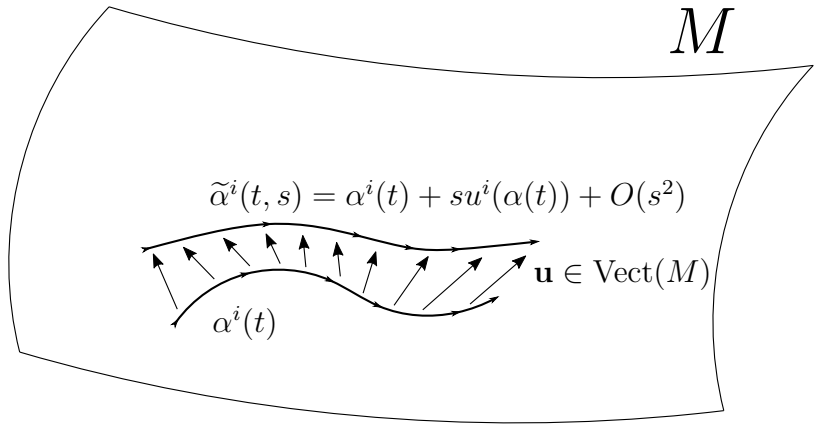}}%
    \put(0.81677297,0.46368427){\color[rgb]{0,0,0}\makebox(0,0)[lt]{\lineheight{1.25}\smash{\begin{tabular}[t]{l}\Huge$M$\end{tabular}}}}%
    \put(0.24943377,0.27332322){\color[rgb]{0,0,0}\makebox(0,0)[lt]{\lineheight{1.25}\smash{\begin{tabular}[t]{l}\large$\widetilde{\alpha}^i(t,s) = \alpha^i(t) + su^i(\alpha(t)) + O(s^2)$\end{tabular}}}}%
    \put(0.65145905,0.17718139){\color[rgb]{0,0,0}\makebox(0,0)[lt]{\lineheight{1.25}\smash{\begin{tabular}[t]{l}\large$\mathbf{u}\in\mathrm{Vect}(M)$\end{tabular}}}}%
    \put(0.25343377,0.11552557){\color[rgb]{0,0,0}\makebox(0,0)[lt]{\lineheight{1.25}\smash{\begin{tabular}[t]{l}\large$\alpha^{i}(t)$\end{tabular}}}}%
  \end{picture}%
\endgroup%

\caption{The effect of the flow of $\mathbf{u}\in\mathrm{Vect}(M)$ on the path $\alpha^i(t)$.}\label{fig:completelift}
\end{figure}

Now consider how the $\stackrel{(1)}{\sim}$-equivalence class of $\widetilde{\alpha}^i(t,s)$, i.e. its value at $t=0$ and its first $t$-derivative at $t=0$, changes with $s$. Let
\begin{align}
    \widetilde{x}^i(s) &= \widetilde{\alpha}^i(0,s),& x^i &= \widetilde{x}^i(0), \\
    \widetilde{y}^i(s) &= \frac{\partial\widetilde{\alpha}^i}{\partial t}(0,s),& y^i &= \widetilde{y}^i(0).
\end{align}
Then
\begin{align}
    \widetilde{x}^i(s) &=& &\alpha^i(0) + su^i(\alpha(0)) + O(s^2) &=& &&x^i + su^i(x) + O(s^2), \\
    \widetilde{y}^i(s) &=& &\frac{\mathrm{d}\alpha^i}{\mathrm{d} t}(0) + s \frac{\partial u^i}{\partial x^j}(\alpha(0))\frac{\mathrm{d}\alpha^j}{\mathrm{d} t}(0) + O(s^2) &=& &&y^i + s \frac{\partial u^i}{\partial x^j}(\alpha(0)) y^j + O(s^2).
\end{align}
By taking the order $s$ terms in the above equations, we have obtained the complete lift of a vector field $\mathbf{u}$ on $M$ to a vector field $\mathbf{u}^{\#}$ on the tangent bundle $TM$. The components of $\mathbf{u}^{\#}$ with respect to the coordinate system $(x^i,y^i)$ are
\begin{align}
 \left(u^i(x),\frac{\partial u^i}{\partial x^j}(x) y^j\right). 
\end{align}
Now the associated differential operator for $\mathbf{u}^{\#}\in\mathrm{Vect}(TM)$ that acts on smooth functions on $TM$ will be
\begin{equation}
u^i(x)\frac{\partial}{\partial x^i} + \frac{\partial u^i}{\partial x^j}(x)y^j \frac{\partial}{\partial y^i}.
\end{equation}
The crucial property of complete lifts is that the assignment $\mathbf{u}\mapsto\mathbf{u}^{\#}$ is a \emph{Lie algebra homomorphism} $\mathrm{Vect}(M)\rightarrow\mathrm{Vect}(TM)$ with respect to the commutator bracket of vector fields on $M$ and $TM$ respectively. Symbolically this is
\begin{align}\label{homo1}
    \left([\mathbf{u},\mathbf{v}]_M\right)^{\#} = \left[\mathbf{u}^{\#},\mathbf{v}^{\#}\right]_{TM},
\end{align}
which can be checked by direct computation or otherwise. (The subscripts indicate the space on which the vector fields live.) A proof sketch can be found in appendix \textbf{\ref{app:diff}}.

This immediately implies that vector fields $\mathbf{u}\in\mathrm{Vect}(M)$ on $M$ can act on functions $f\in C^\infty(TM)$ on $TM$ via $f\mapsto \mathbf{u}^{\#}\cdot f$, and moreover this action is a Lie algebra representation:
\begin{align}
\left[\mathbf{u}^{\#},\mathbf{v}^{\#}\right]_{TM} \cdot f = \mathbf{u}^{\#}\cdot\mathbf{v}^{\#}\cdot f - \mathbf{v}^{\#}\cdot\mathbf{u}^{\#}\cdot f = \left([\mathbf{u},\mathbf{v}]_M\right)^{\#} \cdot f.
\end{align}
Thus we can add this representation of $\mathrm{Vect}(M)$ on $C^\infty(TM)$ to the semidirect product $\mathrm{Vect}(M)\ltimes C^\infty(M)$ relevant to ideal compressible fluid dynamics to obtain the Lie algebra
\begin{align}
    \mathfrak{g}_s = \mathrm{Vect}(M)\ltimes \left(C^\infty(M) \oplus C^\infty(TM) \right),
\end{align}
which will be relevant to describing suspensions of bead-spring pairs in compressible fluids.

\subsection{The $\psi$-subbracket for the distribution function of bead-spring pairs}\label{psisub1}
Now we take the dual $\mathfrak{g}_s^{*}$ of $\mathfrak{g}_s = \mathrm{Vect}(M)\ltimes \left(C^\infty(M) \oplus C^\infty(TM) \right)$ and consider the Lie--Poisson dynamics for functionals on $\mathfrak{g}_s^{*}$. As a vector space, $\mathfrak{g}_s^{*}$ consists of elements of the form $\left({m}_i, \rho, \psi\right)$, where ${m}_i\mathrm{d}^n{x}$ is the momentum density of the fluid, $\rho \mathrm{d}^n{x}$ is the mass density of the fluid, and $\psi(x,y)\mathrm{d}^n{x}\mathrm{d}^n{y}$ can be interpreted as the number density of bead-spring pairs with centre located at $x^i$ and a relative displacement vector $y^i$ between beads. The dual pairing between $C^\infty(TM)$ and $\Omega^{2n}(TM)$ is given by integration over $TM$ i.e. over all $x, y$.

The minus Lie--Poisson bracket for functionals $F[\mathbf{m},\rho,\psi], G[\mathbf{m},\rho,\psi]$ will be
\begin{align} \label{psibracket0}
    \{F,G\}[\mathbf{m},\rho,\psi] =& -\left\langle \mathbf{m},\left[\frac{\delta F}{\delta \mathbf{m}},\frac{\delta G}{\delta \mathbf{m}}\right] \right\rangle - \left\langle \rho, \frac{\delta F}{\delta \mathbf{m}} \cdot \frac{\delta G}{\delta \rho} - \frac{\delta G}{\delta \mathbf{m}} \cdot \frac{\delta F}{\delta \rho} \right\rangle  - \left\langle \psi, \left(\frac{\delta F}{\delta \mathbf{m}}\right)^{\#} \cdot \frac{\delta G}{\delta \psi} - \left(\frac{\delta G}{\delta \mathbf{m}}\right)^{\#} \cdot \frac{\delta F}{\delta \psi} \right\rangle, \nonumber \\
    =& \{F,G\}_{fluids} + \{F,G\}_{\psi},
\end{align}
where we have separated the usual fluid bracket and the $\psi$-subbracket as follows:
\begin{align} 
\{F,G\}_{fluids} &= -\left\langle \mathbf{m},\left[\frac{\delta F}{\delta \mathbf{m}},\frac{\delta G}{\delta \mathbf{m}}\right] \right\rangle - \left\langle \rho, \frac{\delta F}{\delta \mathbf{m}} \cdot \frac{\delta G}{\delta \rho} - \frac{\delta G}{\delta \mathbf{m}} \cdot \frac{\delta F}{\delta \rho} \right\rangle, \\
\label{psibracket1}
\{F,G\}_{\psi} &= - \left\langle \psi, \left(\frac{\delta F}{\delta \mathbf{m}}\right)^{\#} \cdot \frac{\delta G}{\delta \psi} - \left(\frac{\delta G}{\delta \mathbf{m}}\right)^{\#} \cdot \frac{\delta F}{\delta \psi} \right\rangle .
\end{align}
Now focus on $M = \mathbb{R}^n$ and consider Hamiltonians of the form
\begin{align}
    H = H_{fluids}[\mathbf{m},\rho] + H_s[\psi],
\end{align}
where $H_{fluids}$ is the usual ideal compressible fluid Hamiltonian (\ref{Hfluids}), and $H_s[\psi]$ is some internal (free) energy of the bead-spring pairs, which is unspecified for now.

The coordinate expression of the subbracket (\ref{psibracket1}) conincides with the $\psi$-subbracket in \cite{Grmela88}, obtained from direct inspection. In fact, letting $u^i = {\delta H}/{\delta m_i}$, the terms in the subbracket $\{F,H\}_{\psi}$ are explicitly
\begin{align}
    \{F,H\}_{\psi} =& - \int\mathrm{d}^n{x}\mathrm{d}^n{y} \ \psi \left[ \frac{\delta F}{\delta m_i}\frac{\partial}{\partial x^i}\left(\frac{\delta H_s}{\delta \psi}\right) + \frac{\partial}{\partial x^j}\left(\frac{\delta F}{\delta m_i}\right)y^j \frac{\partial}{\partial y^i}\left(\frac{\delta H_s}{\delta \psi}\right) \right] \nonumber \\
    & \qquad - \int\mathrm{d}^n{x}\mathrm{d}^n{y} \ \psi \left[ - u^i\frac{\partial}{\partial x^i}\left(\frac{\delta F}{\delta \psi}\right) - \frac{\partial u^i}{\partial x^j}y^j \frac{\partial}{\partial y^i}\left(\frac{\delta F}{\delta \psi}\right) \right], \nonumber \\
    =& \int\mathrm{d}^n{x}\mathrm{d}^n{y} \ \frac{\delta F}{\delta m_i} \left[ -\psi \frac{\partial}{\partial x^i}\left(\frac{\delta H_s}{\delta \psi}\right) + \frac{\partial}{\partial x^j}\left( \psi y^j \frac{\partial}{\partial y^i} \left(\frac{\delta H_s}{\delta \psi}\right) \right)\right] \nonumber \\
    & \qquad + \int\mathrm{d}^n{x}\mathrm{d}^n{y} \  \frac{\delta F}{\delta \psi}\left[ -\frac{\partial}{\partial x^i}\left(u^i \psi \right) - \frac{\partial}{\partial y^i}\left(\frac{\partial u^i}{\partial x^j}y^j \psi\right) \right], \nonumber \\
    =& \left\langle \frac{\delta F}{\delta m_i} , \int \mathrm{d}^n y \ \left[ -\psi \frac{\partial}{\partial x^i}\left(\frac{\delta H_s}{\delta \psi}\right) + \frac{\partial}{\partial x^j}\left( \psi y^j \frac{\partial}{\partial y^i} \left(\frac{\delta H_s}{\delta \psi}\right) \right) \right] \right\rangle \nonumber \\
    & \qquad + \left\langle \frac{\delta F}{\delta \psi}, -\frac{\partial}{\partial x^i}\left(u^i \psi \right) - \frac{\partial}{\partial y^i}\left(\frac{\partial u^i}{\partial x^j}y^j \psi\right) \right\rangle ,
\end{align}
from which we obtain
\begin{align}
    \label{psieom1}
    \dt{\psi} &= -\frac{\partial}{\partial x^i}\left(u^i \psi \right) - \frac{\partial}{\partial y^i}\left(\frac{\partial u^i}{\partial x^j}y^j \psi\right), \\
    \label{force1}
    \mathcal{F}_i &= \int \mathrm{d}^n y \ \left[ -\psi \frac{\partial}{\partial x^i}\left(\frac{\delta H_s}{\delta \psi}\right) + \frac{\partial}{\partial x^j}\left( \psi y^j \frac{\partial}{\partial y^i} \left(\frac{\delta H_s}{\delta \psi}\right) \right) \right],
\end{align}
where $\mathcal{F}_i$ is the extra force on the fluid due to the suspension.

The Lie--Poisson bracket gives the terms in the \enquote{Liouville equation} for the distribution function of bead-spring pairs due to beads advecting with the flow like Lagrangian markers (see \cite{Bird87,Renardy00}). Note that:
\begin{enumerate}
\item If there are no bead-bead interactions and there is no diffusion, the system is Hamiltonian, because the motion of Lagrangian marker particles in a fluid is time-reversible. It happens that a sufficiently small spherical bead in a Stokes flow is well-approximated by a Lagrangian marker. This is the reason behind the fact that the \enquote{advective} part of the Liouville equation for the distribution function can be captured by the semidiriect product Lie--Poisson formulation.

However, while the motion of a Lagrangian marker in Stokes flow is time-reversible, the motion of an elastic body, e.g. a Hookean spring, in a Stokes flow is \emph{not} time reversible, We will need a piece of \emph{dissipation bracket} to describe the effect of the internal elastic forces within the body.
\item The semidirect product Lie--Poisson system implements the following physical principles in a convenient manner:
\begin{itemize}
\item Energy is conserved: $\dt{H} = \{H,H\} = 0$.
\item The bead-spring pairs evolve in time as Lagrangian marker particles, in the sense made precise by the construction of $\mathbf{u}^{\#}$ in section \textbf{\ref{subsec:liftTM}}.
\end{itemize}
In principle, from these two facts, one can deduce the force of the bead-spring pairs on the fluid by invoking energy conservation and Newton's third law, since we know how the fluid acts on the suspension. If done from the equations of motion directly, this often involves a lot of uninformative manipulations. However, the semidirect product Lie--Poisson formulation provides an expedient way for the force $\mathcal{F}_i$ on the fluid by the suspension to be calculated from the energy $H_s[\psi]$, via direct manipulation of the terms in the subbracket due to $\psi$. This has been emphasised in \cite{Grmela89} and applied in \cite{Beris94book} as a uniform way to derive expressions for stress tensors for different semidirect product Lie--Poisson structures.
\end{enumerate}
For our purposes, since $\mathcal{F}_i$ is a $y$-integral against $\psi$, we will find that for reasonable choices of $H_s[\psi]$, the force can be written as the divergence of a stress tensor, and that the stress tensor depends on \emph{$y$-moments of $\psi$} i.e. combinations $\int \mathrm{d}^ny \  p(y)\psi$ for polynomials $p(y)$.

\subsection{The dissipation bracket and double bracket dynamics}\label{subsec:disbra}
To extend the Lie--Poisson dynamical system to dissipative systems, it is customary to include an additional symmetric positive semidefinite \emph{dissipation bracket} $(\cdot,\cdot)$, such that
\begin{align}
    \dt{F} = \{F,H\} - \frac{1}{\zeta}(F,H),
\end{align}
where $1/\zeta > 0$ is the \emph{mobility} parameter. For our purposes, the dissipation bracket can be thought of as an (approximate) implementation of the mobility relations in a Stokes flow, which relate the velocity of a bead relative to the fluid around it to any additional forces on the bead.

In this formulation $H$ must be interpreted as a \emph{free energy}, since
\begin{align}
    \dt{H} = - \frac{1}{\zeta}(H,H) \leq 0 .
\end{align}
There is no obvious geometric justification for the dissipation bracket, in contrast to the Poisson bracket. Attempts to justify dissipation brackets have mainly been motivated by thermodynamic principles, such that the double bracket formulation is a model for non-equilibrium thermodynamics. Discussion on double bracket dynamics can be found in, for example, \cite{Morrison86, Beris94book}. A recent review in the context of viscoelastic fluids can also be found in \cite{Mackay19}.

It is common to require the bracket to be \emph{bilinear} and satisfy the \emph{Leibniz/product rule}:
\begin{align}
    (FG,H) = F(G,H) + (F,H)G ,
\end{align}
whenever the terms are well-defined. These requirements allows us to write the dissipation bracket of two functionals in a similar manner to (\ref{chainrule}). If $F,H$ are two functionals on the vector space $V$, the bracket $(F,H)$ can be written as
\begin{align}
(F,H)[w] = \left\langle \frac{\delta F}{\delta v} , K[w] \left( \frac{\delta H}{\delta v}[w] \right) \right\rangle,
\end{align}
where $w$ is an arbitrary element in $V$, and $K[w] : V^* \rightarrow V$ is some linear operator that varies with $w$. By a similar argument used in (\ref{chainrule}), the trajectory of a point $w\in V$ is now given in terms of the Hamiltonian functional $H$, the Poisson tensor $J$, and the new linear operator $K$ as
\begin{align}
\dt{w} = J[w] \left(\frac{\delta H}{\delta v}[w]\right) + K[w] \left(\frac{\delta H}{\delta v}[w]\right).
\end{align}

We will now describe a dissipation bracket that gives the correct evolution equation for the distribution function $\psi$. If $M = \mathbb{R}^n$, then $TM$ can be thought of as $\mathbb{R}^{2n}$ with coordinates $(x^i,y^i)$. Let $\delta_{ij}$ denote the standard metric on $\mathbb{R}^n$. Then we can introduce a metric on $TM$ given by
\begin{align}
    \mathrm{d}s^2 = \delta_{ij}\mathrm{d}x^i\mathrm{d}x^j + \delta_{ij}\mathrm{d}y^i\mathrm{d}y^j .
\end{align}
This is positive definite. The generalisation of this metric to general tangent bundles $TM$ of a Riemannian manifold $M$ is called the \emph{Sasaki metric} as studied in \cite{Yano73book}. (See appendix \textbf{\ref{app:vect}} for more details.) Note that in principle, for any $\lambda >0$
\begin{align}\label{ds^2lambda}
\mathrm{d}s^2(\lambda) = \delta_{ij}\mathrm{d}x^i\mathrm{d}x^j + \lambda \delta_{ij}\mathrm{d}y^i\mathrm{d}y^j,
\end{align}
is also a valid Riemannian metric on $TM$. We will choose $\lambda = 1/2$ so that the equations we derive later coincide with those in \cite{Renardy00, Bird87}. This is equivalent to redefining the $y$-coordinate by scaling -- the hydrodynamic part of the evolution equation for $\psi$ is invariant under rescaling of the $y$-coordinate. In rheological applications, it is typical to consider the bead-spring length scale $\lvert y\rvert$ to be much shorter than the flow length scale $\lvert x\rvert$, and we can adjust the factors in the metric as appropriate to reflect this.

To mimic the Lie--Poisson bracket as closely as possible, consider
\begin{align}
(F,G) = \left\langle \psi, \widetilde{g}\left( \mathrm{d} \left(\frac{\delta F}{\delta \psi}\right), \mathrm{d} \left(\frac{\delta G}{\delta \psi}\right) \right) \right\rangle,
\end{align}
where $\widetilde{g}$ is the inverse of the metric (\ref{ds^2lambda}) with $\lambda = 1/2$ on $TM$. In coordinates:
\begin{align}
    (F,G) =& \int \mathrm{d}^n x \mathrm{d}^n y \ \psi \left[ \delta^{ij} \frac{\partial}{\partial x^i}\left( \frac{\delta F}{\delta \psi} \right) \frac{\partial}{\partial x^j}\left( \frac{\delta G}{\delta \psi} \right) + 2\delta^{ij} \frac{\partial}{\partial y^i}\left( \frac{\delta F}{\delta \psi} \right) \frac{\partial}{\partial y^j}\left( \frac{\delta G}{\delta \psi} \right) \right] \nonumber \\
    =& - \left\langle \frac{\delta F}{\delta \psi}, \frac{\partial}{\partial x^i}\left( \delta^{ij}\psi \frac{\partial}{\partial x^j}\left(\frac{\delta G}{\delta \psi}\right) \right) + 2 \frac{\partial}{\partial y^i}\left( \delta^{ij}\psi \frac{\partial}{\partial y^j}\left(\frac{\delta G}{\delta \psi}\right) \right) \right\rangle.
\end{align}
This type of dissipation bracket coming from a metric (or indeed a more complicated quadratic form) allows us to write the evolution of $\psi$ as a \enquote{Liouville equation} $\dt{\psi} + \nabla_x \cdot(\dt{x}\psi) + \nabla_y \cdot(\dt{y}\psi) = 0$, since the extra terms it produces can be factored into some \enquote{flow velocity} in configuration space, proportional to the gradient of the free energy per bead-spring pair ${\delta H}/{\delta \psi}$.

For $H = H_{fluids}[\mathbf{m},\rho] + H_s[\psi]$, the equation of motion for $\psi$ is
\begin{align} \label{psieom1dis}
    \dt{\psi} +& \frac{\partial}{\partial x^i}\left(u^i \psi \right) + \frac{\partial}{\partial y^i}\left(\frac{\partial u^i}{\partial x^j}y^j \psi\right) = \frac{1}{\zeta}\left[ \frac{\partial}{\partial x^i}\left( \delta^{ij}\psi \frac{\partial}{\partial x^j}\left(\frac{\delta H_s}{\delta \psi}\right) \right) + 2 \frac{\partial}{\partial y^i}\left( \delta^{ij}\psi \frac{\partial}{\partial y^j}\left(\frac{\delta H_s}{\delta \psi}\right) \right) \right] ,
\end{align}
and $H_s[\psi]$ should now be interpreted as the free energy associated with the bead-spring pairs.

The dissipation bracket implements bead-bead interactions, since it couples $H_s[\psi]$ to $\psi$, but does not add any extra coupling between $\mathbf{m}$ and $\psi$. As there are no extra fluid-bead couplings, it is still valid to calculate the force on the fluid due to the suspension with the semidirect product Lie--Poisson formulation.

With this dissipation bracket, the kinematics of the dissipation mechanism is that the dissipative forces are generated by gradients of the free energy, and the mobility relation that converts forces to velocities is linear.

One can implement more sophisticated mobility relations by modifying the dissipation bracket \cite{Beris94book}. One possible effect to include is the \emph{hydrodynamic interaction} between the beads. This can be done by using a slightly different dissipation bracket. For $n=3$, let
\begin{align}
(F,G) =& \int \mathrm{d}^3 x \mathrm{d}^3 y \ \psi \left[ \delta^{ij} \frac{\partial}{\partial x^i}\left( \frac{\delta F}{\delta \psi} \right) \frac{\partial}{\partial x^j}\left( \frac{\delta G}{\delta \psi} \right) + 2D^{ij}(y) \frac{\partial}{\partial y^i}\left( \frac{\delta F}{\delta \psi} \right) \frac{\partial}{\partial y^j}\left( \frac{\delta G}{\delta \psi} \right) \right],
\end{align}
where $D^{ij}(y)$ is related to the \emph{Oseen--Burgers tensor} $ \Omega^{ij} = {(|y|^2\delta^{ij}+y^i y^j)}/{|y|^3}$, where $|y| = \sqrt{\delta_{ij}y^iy^j}$, as follows. For some parameter $\gamma$,
\begin{align}
    D^{ij}(y) = \delta^{ij} + {\gamma}\Omega^{ij}.
\end{align}
The Oseen--Burgers tensor $\Omega^{ij}$ can be used to describe the disturbance flow field around a small sphere in a Stokes flow \cite{Kim05book}. The physical effect included here is the advection of one bead by the disturbance flow field of the other, and vice versa.

Unfortunately, this bracket poses a considerable complication to the theory, since it is not known how to obtain a reduced system of equations for $y$-moments of $\psi$ that are sufficient to describe the stress state, even for Hookean springs. This property is commonly called (finite and exact) \emph{closure} \cite{Renardy00}. Ad hoc techniques such as \enquote{pre-averaging} have been used to obtain an effective constant mobility \cite{Bird87}.

Another piece of physics that the dissipation bracket can implement is the \emph{Newtonian viscous stress} on the fluid. This can be done by including an additional term in the dissipation bracket which is bilinear positive semidefinite in ${\partial}/{\partial x^i}({\delta F}/{\delta m_j})$ and ${\partial}/{\partial x^i}({\delta G}/{\delta m_j})$ with no $\psi$-dependence \cite{Morrison84,Beris90,Edwards90,Beris94book}. (The effect that would be modelled by allowing $\psi$-dependence would be \emph{suspension-enhanced} and possibly \emph{anisotropic viscosity} due to the suspended bodies.) 

It is typical in rheological applications to consider the bead-spring pairs to be suspended in a Newtonian fluid, so that the mechanism responsible for the motion of the beads is Stokes drag. The \emph{elastic particle-contributed stress}, which is the back-reaction on the fluid from doing work to the bead-spring pairs, is energy conserving. The purely dissipative Newtonian viscous stress, which comes from the friction between adjacent fluid parcels, is a separate effect and does not affect the particle-contributed stress. Since we will focus on the evolution equation for the distribution function $\psi$ and the elastic particle-contributed stress for the rest of the paper, we will omit the Newtonian viscous stress term, under the knowledge that it can be added back without modifying any of the expressions obtained.

We can also justify the omission of the Newtonian viscous stress term by a scaling argument. If the lengthscales of the individual bead-spring pairs are sufficiently small, we can model them as being immersed in a Stokes flow; if the macroscopic lengthscales of the fluid are sufficiently large, the Newtonian viscous stress is negligible compared to the other terms in the momentum equation. There is no conflict between these two asymptotic regimes in principle. If we take these limits simultaneously, the resulting system describes the coupling between an ideal compressible fluid and a distribution function of the bead-spring pairs that evolves in a non-conservative manner. This is the type of system that we will consider for the rest of the paper. 

In the next section, we will show that the upper-convected Maxwell model can be derived from the double bracket dynamics of a Hookean bead-spring pair suspension. In this case, including the Newtonian viscous stress term amounts to generalising the upper-convected Maxwell model to the Oldroyd-B model.

\subsection{The upper-convected Maxwell model from double bracket dynamics}\label{subsec:maxwell}

So far we have considered the Poisson and dissipation brackets that are suitable for describing the dynamics of a bead-spring pair suspension. Now we specialise to certain forms of the Hamiltonian functional, and show that the upper-convected Maxwell model can be derived from this double bracket system. Consider an isothermal fluid, with its temperature held at a fixed constant $T$ throughout. As before, we postulate a Hamiltonian functional of the form $H = H_{fluids}[\mathbf{m},\rho] + H_s[\psi]$, where $H_{fluids}$ is the usual ideal compressible fluid Hamiltonian (\ref{Hfluids}), and $H_s[\psi]$ is the free energy of the bead-spring pairs, given by 
\begin{align}\label{freeE1}
H_s[\psi] = \int \mathrm{d}^n x \mathrm{d}^n y \ \left(E(y)\psi + k_BT\psi \log (\psi) \right),
\end{align}
where $k_B$ is the Boltzmann constant, and $E(y)$ is the internal energy of a bead-spring pair with bead-to-bead displacement $y$. In terms of thermodynamics, this is a free energy $U - TS$, where
\begin{align}
U[\psi] = \int \mathrm{d}^n x \mathrm{d}^n y \ E(y)\psi, \quad S[\psi]= -k_B \int \mathrm{d}^n x \mathrm{d}^n y \  \psi \log (\psi) 
\end{align}
are the \emph{internal energy} and the \emph{(Boltzmann) entropy} of the distribution function $\psi$, respectively.

Note that $\log(\psi)$ is only ever defined up to a constant, since we are actually comparing $\psi \mathrm{d}^n x \mathrm{d}^n y$ to the standard volume element $\mathrm{d}^n x \mathrm{d}^n y$ in $\mathbb{R}^{2n}$, which is unique up to scaling. If we replace $\mathrm{d}^n x \mathrm{d}^n y$ with $\Lambda \mathrm{d}^n x \mathrm{d}^n y$ for some positive $\Lambda$, then $\log(\psi) \mapsto \log(\psi/\Lambda) = \log(\psi) -\log(\Lambda)$, and the Boltzmann entropy changes by
\begin{align}
S[\psi] \mapsto S[\psi] + k_B\log(\Lambda) \int \mathrm{d}^n x \mathrm{d}^n y \  \psi .
\end{align}
Fortunately, the functional $N = \int \mathrm{d}^n x \mathrm{d}^n y \  \psi$ is a \emph{Casimir functional} of the Lie--Poisson bracket, since ${\delta N}/{\delta \mathbf{m}} = 0, {\delta N}/{\delta \rho} = 0, {\delta N}/{\delta \psi} = 1$, so for any functional $F$,
\begin{align}
\{F,N\} = - \left\langle \psi, \left(\frac{\delta F}{\delta \mathbf{m}}\right)^{\#} \cdot 1 \right\rangle = 0 .
\end{align}
In addition, since $\mathrm{d}\left( {\delta N}/{\delta \psi} \right) = 0$, $(F,N) = 0$ for all functionals $F$ as well. So the dynamics are unaltered by such a rescaling.

Now we can calculate the functional derivative of the free energy $H_s[\psi]$ in (\ref{freeE1}):
\begin{align}
    \frac{\delta H_s}{\delta \psi} = E(y) + k_B T \left(\log(\psi)+1\right).
\end{align}
Substituting this into the equation for the force on the fluid gives 
\begin{align}
\mathcal{F}_i = \frac{\partial}{\partial x^i} \left(-2k_B T \int \mathrm{d}^n y \ \psi \right) + \frac{\partial}{\partial x^j} \left( \int \mathrm{d}^n y \ y^j\frac{\partial E}{\partial y^i} \psi \right) = \frac{\partial\sigma^j_i}{\partial x^j},
\end{align}
where
\begin{equation}
\sigma^j_i = -2k_B T \int \mathrm{d}^n y \ \psi \delta^j_i + \int \mathrm{d}^n y \ y^j\frac{\partial E}{\partial y^i} \psi 
\end{equation}
is the \emph{stress tensor} exerted on the fluid by the bead-spring pairs. The fact that the force can be written as the divergence of a stress tensor implies that linear momentum is conserved. This is related to the translational invariance of the free energy $H_s[\psi]$ in (\ref{freeE1}). The evolution equation (\ref{psieom1dis}) for $\psi$ becomes a \emph{Fokker--Planck equation}:
\begin{align}\label{psi1eomfinal}
\dt{\psi} + & \frac{\partial}{\partial x^i}\left(u^i \psi \right) + \frac{\partial}{\partial y^i}\left(\frac{\partial u^i}{\partial x^j}y^j \psi\right) = \frac{2}{\zeta}  \frac{\partial}{\partial y_i} \left(\delta^{ij}\frac{\partial E}{\partial y^j} \psi \right) + \frac{k_B T}{\zeta}\left( \nabla^2_x \psi + 2 \nabla^2_y \psi \right),
\end{align}
where $\nabla^2_x = \delta^{ij}{\partial^2}/{\partial x^i \partial x^j}$ and $\nabla^2_y = \delta^{ij}{\partial^2}/{\partial y^i \partial y^j}$.
This coincides with the kinetic equation for bead-spring pairs with arbitrary internal energy $E(y)$ obtained in \cite{Bird87, Renardy00}, from a double bracket formulation.

There is no known general procedure to obtain closed evolution equations for the relevant moments of $\psi$ appearing in the stress tensor for general $E(y)$. However, this is possible if the energy is quadratic, i.e. $E(y) = (\kappa /2)\delta_{ij}y^iy^j$, which will reproduce the upper-convected Maxwell model. Attempts in obtaining closed evolution equations for general $E(y)$ include the \emph{Peterlin approximation}, which approximates the energy $E(y)$ as that of a Hookean spring, with an effective spring constant depending on the second $y$-moment of $\psi$ \cite{Bird87, Renardy00}. This approach produces nonlinear, but closed, evolution equations for the relevant moments of $\psi$.

Writing $\sigma^{jk} = \sigma^j_i \delta ^{ik}$ for convenience, the stress tensor is
\begin{align} \label{model1stress}
\sigma^{jk} = -2k_B T \delta^{jk} \int \mathrm{d}^ny \ \psi + \kappa \int \mathrm{d}^n y \ y^j y^k \psi = -2k_BT\delta^{jk} \langle 1 \rangle + \kappa \left\langle y^j y^k \right\rangle,
\end{align}
where $\left\langle \cdots \right\rangle = \int \mathrm{d}^ny \ \psi \left(\cdots \right)$. Note that $\left\langle 1 \right\rangle = \int \mathrm{d}^ny \ \psi = n(x)$ is the number density of bead-spring pairs in real space, and so it is not surprising that each bead in a pair contributes $n(x)k_BT$ to the isotropic pressure.

The evolution equation for the distribution function $\psi$ is then
\begin{align} \label{uppermaxwell1}
\dt{\psi} + & \frac{\partial}{\partial x^i}\left(u^i \psi \right) + \frac{\partial}{\partial y^i}\left(\frac{\partial u^i}{\partial x^j}y^j \psi\right) = \frac{2\kappa}{\zeta}  \frac{\partial}{\partial y^i} \left(y^i \psi \right) + \frac{k_B T}{\zeta}\left( \nabla^2_x \psi + 2 \nabla^2_y \psi \right).
\end{align}
Taking the $1$ and $y^jy^k$ moments of (\ref{uppermaxwell1}) gives
\begin{align}
\label{uppermaxwelln}
\dt{\left\langle 1 \right\rangle} + \frac{\partial}{\partial x^i}\left(u^i \left\langle 1 \right\rangle\right) =& \frac{k_BT}{\zeta}\nabla^2_x \left\langle 1 \right\rangle, \\
\label{uppermaxwellC}
\dt{\left\langle y^j y^k \right\rangle} + \frac{\partial}{\partial x^i}\left( u^i \left\langle y^j y^k \right\rangle\right) - \frac{\partial u^j}{\partial x^l}\left\langle y^l y^k \right\rangle - \frac{\partial u^k}{\partial x^l}\left\langle y^j y^l \right\rangle =& -\frac{4\kappa}{\zeta}\left\langle y^j y^k \right\rangle  + \frac{4k_BT}{\zeta}\delta^{jk}\left\langle 1 \right\rangle + \frac{k_BT}{\zeta}\nabla^2_x \left\langle y^j y^k \right\rangle.
\end{align} 
We can derive the upper-convected Maxwell model from (\ref{uppermaxwelln},\ref{uppermaxwellC}) as follows. Let $n = \left\langle 1 \right\rangle$ be the number density of bead-spring pairs, and $C^{jk} = \left\langle y^j y^k \right\rangle$ be the \emph{conformation tensor}. Note that the stress tensor in (\ref{model1stress}) can be decomposed into two parts: $\sigma^{jk} = \sigma_{(0)}^{jk} + \sigma_{(1)}^{jk}$, where $\sigma_{(0)}^{jk} = -nk_BT\delta^{jk}$ is the pressure of suspended particles without internal structure. The stress due to internal structure $\sigma_{(1)}^{jk} = -nk_BT\delta^{jk} + \kappa C^{jk} $ evolves according to the following equation:
\begin{align}
\dt{\sigma}_{(1)}^{jk} &+ \frac{\partial}{\partial x^i}\left( u^i \sigma_{(1)}^{jk}\right) - \frac{\partial u^j}{\partial x^l}\sigma_{(1)}^{lk} - \frac{\partial u^k}{\partial x^l}\sigma_{(1)}^{jl} \nonumber = -\frac{4\kappa}{\zeta}\sigma_{(1)}^{jk}  + nk_B T \left(\frac{\partial u^j}{\partial x^l}\delta^{lk} + \frac{\partial u^k}{\partial x^l}\delta^{jl}  \right) + \frac{k_BT}{\zeta}\nabla^2_x\sigma_{(1)}^{jk},
\end{align}
which is the usual upper-convected Maxwell model, with an additional diffusion term. The diffusion term is usually absent because derivations often assume a separation of length scales: $\lvert x\rvert \gg \lvert y \rvert$, i.e. that the macroscopic fluid properties vary on length scales much larger than the length scale of the suspended bodies. We can adjust the factors in the metric (\ref{ds^2lambda}) to make the $x$-diffusion term arbitrarily small relative to the other terms to reflect this separation of scales.

The important point here is that the moments $n = \left\langle 1 \right\rangle,  C^{jk} = \left\langle y^j y^k \right\rangle$ form a closed system of evolution equations, given the fluid flow $\mathbf{u}$, and are sufficient to describe the stress tensor (\ref{model1stress}). This property of \emph{finite closure} allows one to avoid solving for the full distribution function in the configuration space of suspensions.

\section{Higher order tangent bundles and semidirect product Lie algebra set-up for multibead-chains}\label{sec:TNMmodel}

In this section we will describe the construction of the \emph{$N^{th}$ order tangent bundle} from a manifold, which can be considered as the configuration space of a small $(N+1)$-bead chain \cite{Yano73book}. This construction enjoys many analogous properties to the tangent bundle considered in section \textbf{\ref{sec:2bead}}, which allows us to consider the advection of distribution functions for multibead-chains as semidirect product Lie--Poisson system.

\subsection{The higher order tangent bundle $T^{(N)}M$}\label{subsec:defineTNM}
There are better finite-dimensional approximations to the space of all paths $\alpha: I \rightarrow M$, where as before $I$ is some closed interval containing $0$ in its interior. They can be constructed by considering the equivalence relation $\stackrel{(N)}{\sim}$ on paths, defined by the following.

Let $\alpha, \beta$ be paths, and $x^i$ be local coordinates around $\alpha(0)$, so that $\alpha^i(t)$ are the coordinates of the path $\alpha$. Then $\alpha$ is said to \emph{have $N^{th}$ order contact with $\beta$ at $t=0$}, written as $\alpha \stackrel{(N)}{\sim} \beta$, if and only if
\begin{align}
\alpha^i(0) &= \beta^i(0), \nonumber \\
\frac{\mathrm{d}\alpha^i}{\mathrm{d} t}(0) &= \frac{\mathrm{d}\beta^i}{\mathrm{d} t}(0), \nonumber \\
&\vdots  \nonumber \\
\frac{\mathrm{d}^N\alpha^i}{\mathrm{d} t^N}(0) &= \frac{\mathrm{d}^N\beta^i}{\mathrm{d} t^N}(0).
\end{align} 
There are $n(N+1)$ equations to be satisfied. Conversely, given $x^i, y^i_{(1)}, \ldots, y^i_{(N)}$, it is possible to construct a path $\alpha(t)$ with $x^i$ as the coordinates of $\alpha(0)$, and $y^i_{(a)}$ as the coordinates of its $a^{th}$ derivative at $t=0$ for $a=1,\ldots,N$, using coordinates:
\begin{align}
\alpha^i(t) = x^i + \sum_{a=1}^{N} \frac{t^a}{a!}y^i_{(a)} + O(t^{N+1}),
\end{align}
so each of these equivalence classes is nonempty, and the space of all such equivalence classes is called the \emph{$N^{th}$ order tangent bundle} $T^{(N)}M$ of $M$, and is an $n(N+1)$-dimensional manifold. (The $N^{th}$ order tangent bundle is also known as the \emph{space of $N$-jets of $\mathbb{R}$ into $M$ with fixed source} and is commonly denoted by $J^N_0(\mathbb{R},M)$.) The case $N=1$ gives the usual tangent bundle from section \textbf{\ref{subsec:defineTM}}. We will use $(x^i, y^i_{(1)}, \ldots, y^i_{(N)})$ to denote the coordinates of a point in $T^{(N)}M$.

Informally, this is like an $(N+1)$-point approximation to a small segment of a curve attached to $M$ by an $N^{th}$ order polynomial. Note however that this is a local approximation to a smooth path based at a point, and knowing the full Taylor series of a smooth path at one point is not in general sufficient to determine the value of the path at another arbitrarily close point (for example, the functions $f(t) = \exp(-1/t^2)$ and $f(t)=0$ have the same Taylor series at $t=0$).

\subsection{Complete lifts of vector fields to $T^{(N)}M$}\label{liftTNM}

By adapting the same computation for how vector fields act on $TM$, we can allow vector fields to act on $T^{(N)}M$ in a natural way \cite{Yano73book}.

Let $\mathbf{u}$ be a vector field on $M$. We will construct a vector field $\mathbf{u}^{\#}$ on $T^{(N)}M$, called the \emph{complete lift} of $\mathbf{u}$, as follows. Let $\alpha(t)$ be a path on $M$. The flow of $\mathbf{u}$ by a small parameter $s$ deforms the path to $\widetilde{\alpha}(t,s)$, see figure \ref{fig:completelift}, whose components are
\begin{align}
\widetilde{\alpha}^i(t,s) = \alpha^i(t) + su^i(\alpha(t)) + O(s^2).
\end{align}
Again, let 
\begin{align}
\widetilde{x}^i(s) &= \widetilde{\alpha}^i(0,s), \nonumber \\
\widetilde{y}^i_{(a)}(s) &= \frac{\partial^a}{\partial t^a} \bigg\rvert_{t=0} \widetilde{\alpha}^i(t,s), \quad \text{for $a = 1, \ldots, N$}.
\end{align}
and denote their values at $s=0$ as
\begin{align}
\widetilde{x}^i(0) = x^i, \quad \widetilde{y}^i_{(a)}(0) = y^i_{(a)} \quad \text{for $a = 1, \ldots, N$}.
\end{align}
If we apply $({\partial}/{\partial t})^a \rvert_{t=0}$ to $\widetilde{\alpha}(t,s)$ for $a = 1, \ldots, N$, we obtain the following expressions:
\begin{align}
\widetilde{x}^i(s) &= x^i + su^i(x) + O(s^2), \nonumber \\
\widetilde{y}^i_{(1)}(s) &= y^i_{(1)} + s\frac{\partial u^i}{\partial x^j}(x) y^j_{(1)} + O(s^2), \nonumber \\
\widetilde{y}^i_{(2)}(s) &= y^i_{(1)} + s\left(\frac{\partial^2 u^i}{\partial x^j x^k}(x) y^j_{(1)} y^k_{(1)} + \frac{\partial u^i}{\partial x^j}(x) y^j_{(2)}\right) + O(s^2), \nonumber \\
& \vdots \nonumber \\
\widetilde{y}^i_{(N)}(s) &= y^i_{(N)} + s \left( \frac{\mathrm{d}^N}{\mathrm{d} t^N} u^i(\alpha(t)) \right) \bigg\rvert_{t=0} + O(s^2).
\end{align}
This flow, which is parametrised by $s$, defines a vector field $\mathbf{u}^{\#}$ on $T^{(N)}M$ called the \emph{complete lift}. The components of $\mathbf{u}^{\#}$ relative to the coordinate system $(x^i,y^i_{(1)},\cdots, y^i_{(N)})$ can be obtained from taking $({\partial}/{\partial s} )\rvert_{s=0}$ of the expressions above:
\begin{align}
\mathbf{u}^{\#}(x,y_{(1)},\ldots, y_{(N)}) &= \left(u^i(\alpha(0)), \frac{\mathrm{d}}{\mathrm{d} t} u^i(\alpha(t)) \bigg\rvert_{t=0}, \ldots,   \frac{\mathrm{d}^N}{\mathrm{d} t^N} u^i(\alpha(t)) \bigg\rvert_{t=0} \right), \\
\text{where} \quad \alpha^i(t) &= x^i + \sum_{a=1}^{N} \frac{t^a}{a!}y^i_{(a)} + O(t^{N+1}).
\end{align}
The choice of the $O(t^{N+1})$ term will not affect the components of $\mathbf{u}^{\#}$, which can be seen by Taylor expanding and comparing powers of $t$.

The crucial property here is that the complete lift of vector fields is a Lie algebra homomorphism $\mathrm{Vect}(M)\rightarrow\mathrm{Vect}(T^{(N)}M)$, as
\begin{align}\label{homoN}
\left([\mathbf{u},\mathbf{v}]_M\right)^{\#} = \left[\mathbf{u}^{\#},\mathbf{v}^{\#}\right]_{T^{(N)}M}.
\end{align}
This can be checked by direct computation or otherwise. A sketch of proof is given in appendix \textbf{\ref{app:diff}}.

This means the action of vector fields $\mathbf{u}$ on $M$ on smooth functions $f$ on $T^{(N)}M$ given by $f \mapsto \mathbf{u}^{\#} \cdot f$ is a Lie algebra representation. Hence one can form the corresponding semidirect product with vector fields $\mathbf{u}$ on $M$ acting on functions $f$ on $T^{(N)}M$ by the complete lift.

The dual space to $C^\infty(T^{(N)}M)$ can be thought of as the space of volume forms $\Omega^{n(N+1)}(T^{(N)}M)$, with typical element $\psi \mathrm{d}^nx\mathrm{d}^ny_{(1)}\ldots\mathrm{d}^ny_{(N)}$. Such a $\psi$ can be thought of as a distribution function on the configuration space of $(N+1)$-bead chains.

\subsection{Subbracket for the distribution function in the semidirect product Lie algebra formulation}\label{subsec:Nsubbracket}
The semidirect product Lie algebra relevant to a compressible fluid advecting $(N+1)$-bead chains is
\begin{equation}
\mathfrak{g}_s = \mathrm{Vect}(M)\ltimes \left( C^\infty(M) \oplus C^\infty(T^{(N)}M) \right).
\end{equation}
The dual space $\mathfrak{g}_s^{*}$ is isomorphic to $\mathrm{Vect}(M)^{*} \oplus \Omega^n(M) \oplus \Omega^{n(N+1)}(T^{(N)}M)$ as a vector space. An element of $\mathfrak{g}_s^{*}$ is a triple $(\mathbf{m},\rho,\psi)$, where $\mathbf{m}$ is the fluid momentum density, $\rho$ is the fluid mass density, and $\psi$ is the distribution function of the suspended bead-chains in configuration space.

The minus Lie--Poisson bracket for functionals $F = F[\mathbf{m},\rho,\psi]$ and $G = G[\mathbf{m},\rho,\psi]$ is
\begin{align}
\{F,G\}[\mathbf{m},\rho,\psi] =& -\left\langle \mathbf{m},\left[\frac{\delta F}{\delta \mathbf{m}},\frac{\delta G}{\delta \mathbf{m}}\right] \right\rangle - \left\langle \rho, \frac{\delta F}{\delta \mathbf{m}} \cdot \frac{\delta G}{\delta \rho} - \frac{\delta G}{\delta \mathbf{m}} \cdot \frac{\delta F}{\delta \rho} \right\rangle - \left\langle \psi, \left(\frac{\delta F}{\delta \mathbf{m}}\right)^{\#} \cdot \frac{\delta G}{\delta \psi} - \left(\frac{\delta G}{\delta \mathbf{m}}\right)^{\#} \cdot \frac{\delta F}{\delta \psi} \right\rangle \nonumber \\
=& \{F,G\}_{fluids} + \{F,G\}_{\psi},
\end{align}
where $\{F,G\}_{fluids}$ denotes the usual compressible fluid bracket, and $\{F,G\}_{\psi}$ denotes $\psi$-subbracket i.e. the terms that explicitly involve $\psi$.

\section{The $3$-bead chain model}\label{sec:3beads}

As an application of the ideas of the previous section, consider the case $N=2$. Then $T^{(2)}M$ is the configuration space of a $3$-bead chain. The extra degrees of freedom $y_{(1)} = y, y_{(2)} = z$ (renamed for notational clarity) can be thought of as the average extension and the bending of the chain, respectively, as shown in figure \ref{fig:3bead}.

\begin{figure}
\centering
\begingroup%
  \makeatletter%
  \providecommand\color[2][]{%
    \errmessage{(Inkscape) Color is used for the text in Inkscape, but the package 'color.sty' is not loaded}%
    \renewcommand\color[2][]{}%
  }%
  \providecommand\transparent[1]{%
    \errmessage{(Inkscape) Transparency is used (non-zero) for the text in Inkscape, but the package 'transparent.sty' is not loaded}%
    \renewcommand\transparent[1]{}%
  }%
  \providecommand\rotatebox[2]{#2}%
  \newcommand*\fsize{\dimexpr\f@size pt\relax}%
  \newcommand*\lineheight[1]{\fontsize{\fsize}{#1\fsize}\selectfont}%
  \ifx\svgwidth\undefined%
    \setlength{\unitlength}{294.49954772bp}%
    \ifx\svgscale\undefined%
      \relax%
    \else%
      \setlength{\unitlength}{\unitlength * \real{\svgscale}}%
    \fi%
  \else%
    \setlength{\unitlength}{\svgwidth}%
  \fi%
  \global\let\svgwidth\undefined%
  \global\let\svgscale\undefined%
  \makeatother%
  \begin{picture}(1,0.55038817)%
    \lineheight{1}%
    \setlength\tabcolsep{0pt}%
    \put(0,0){\includegraphics[width=\unitlength]{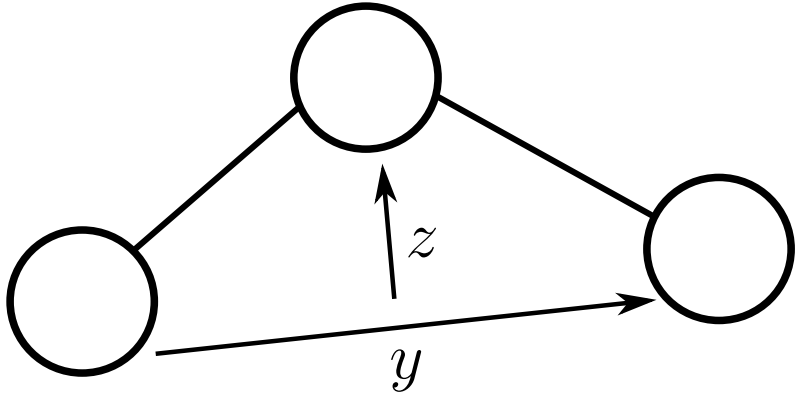}}%
    \put(0.50499129,0.22795843){\color[rgb]{0,0,0}\makebox(0,0)[lt]{\lineheight{1.25}\smash{\begin{tabular}[t]{l}\Huge $z$\end{tabular}}}}%
    \put(0.48490868,0.07278334){\color[rgb]{0,0,0}\makebox(0,0)[lt]{\lineheight{1.25}\smash{\begin{tabular}[t]{l}\Huge $y$\end{tabular}}}}%
  \end{picture}%
\endgroup%

\caption{The coordinates $y$ and $z$ describe the internal degrees of freedom of the $3$-bead chain.}\label{fig:3bead}
\end{figure}

Forming the semidirect product $\mathfrak{g}_s = \mathrm{Vect}(M)\ltimes \left( C^\infty(M) \oplus C^\infty(T^{(2)}M) \right)$ and looking at the Lie--Poisson dynamics for functionals on the dual $\mathfrak{g}_s^*$, we obtain the $\psi$-subbracket as follows:
\begin{align}
\{F,G\}_{\psi} =& - \left\langle \psi, \left(\frac{\delta F}{\delta \mathbf{m}}\right)^{\#} \cdot \frac{\delta G}{\delta \psi} - \left(\frac{\delta G}{\delta \mathbf{m}}\right)^{\#} \cdot \frac{\delta F}{\delta \psi} \right\rangle \nonumber \\
=& - \int \mathrm{d}^nx\mathrm{d}^ny\mathrm{d}^nz \ \psi \Bigg[ \frac{\delta F}{\delta m_i} \frac{\partial}{\partial x^i}\left(\frac{\delta G}{\delta \psi}\right)+ \frac{\partial}{\partial x^j}\left(\frac{\delta F}{\delta m_i}\right)y^j\frac{\partial}{\partial y^i}\left(\frac{\delta G}{\delta \psi}\right), \nonumber \\
& \qquad + \left( \frac{\partial^2}{\partial x^j\partial x^k}\left(\frac{\delta F}{\delta m_i}\right)y^j y^k + \frac{\partial}{\partial x^j}\left(\frac{\delta F}{\delta m_i}\right)z^j \right)\frac{\partial}{\partial z^i}\left( \frac{\delta G}{\delta \psi} \right) \nonumber \\
& \qquad - \frac{\delta G}{\delta m_i} \frac{\partial}{\partial x^i}\left(\frac{\delta F}{\delta \psi}\right)+ \frac{\partial}{\partial x^j}\left(\frac{\delta G}{\delta m_i}\right)y^j\frac{\partial}{\partial y^i}\left(\frac{\delta F}{\delta \psi}\right) \nonumber \\
& \qquad - \left( \frac{\partial^2}{\partial x^j\partial x^k}\left(\frac{\delta G}{\delta m_i}\right)y^j y^k + \frac{\partial}{\partial x^j}\left(\frac{\delta G}{\delta m_i}\right)z^j \right)\frac{\partial}{\partial z^i}\left( \frac{\delta F}{\delta \psi} \right) \Bigg].
\end{align}
If we consider Hamiltonians of the form $H = H_{fluids}[\mathbf{m},\rho] + H_s[\psi]$ and let $\mathbf{u} = {\delta H}/{\delta \mathbf{m}}$, then the evolution equation for $\psi$ and the extra force $\mathcal{F}_i$ on the fluid are 
\begin{align}
\label{psieom2}
\dt{\psi} =& -\frac{\partial}{\partial x^i}\left(u^i \psi \right) - \frac{\partial}{\partial y^i}\left(\frac{\partial u^i}{\partial x^j}y^j \psi\right) - \frac{\partial}{\partial z^i}\left(\left[ \frac{\partial^2 u^i}{\partial x^j \partial x^k}y^j y^k + \frac{\partial u^i}{\partial x^j}z^j \right]\psi \right), \\
\label{force2}
\mathcal{F}_i =& \int \mathrm{d}^n y \mathrm{d}^n z \ \psi \Bigg[ - \frac{\partial}{\partial x^i}\left(\frac{\delta H_s}{\delta \psi}\right) + \frac{\partial}{\partial x^j}\left(  y^j \frac{\partial}{\partial y^i} \left(\frac{\delta H_s}{\delta \psi}\right) \right) \nonumber \\
& \qquad+ \frac{\partial}{\partial x^j}\left(  z^j \frac{\partial}{\partial z^i} \left(\frac{\delta H_s}{\delta \psi}\right) \right) - \frac{\partial^2 }{\partial x^j \partial x^k}\left(y^j y^k \frac{\partial}{\partial z^i} \left(\frac{\delta H_s}{\delta \psi}\right) \right) \Bigg].
\end{align}
This is obtained by the usual procedure of equating
\begin{align}
\dt{F} = \left\langle \frac{\delta F}{\delta \mathbf{m}}, \dt{\mathbf{m}} \right\rangle + \left\langle \frac{\delta F}{\delta \rho}, \dt{\rho} \right\rangle + \left\langle \frac{\delta F}{\delta \psi}, \dt{\psi} \right\rangle = \{F,H\} = \{F,H\}_{fluids} + \{F,H\}_{\psi}
\end{align}
for arbitrary $F$, and collecting the terms proportional to $\delta F / \delta \mathbf{m}$ and $\delta F /\delta \psi$ respectively due to the $\psi$-subbracket, after some integrations by parts.

The term $({\partial^2 u^i}/{\partial x^j \partial x^k})y^j y^k$ in (\ref{psieom2}) can be thought of as the \emph{bending rate} of the bead chain due to a difference between the stretching velocities across the bead chain. Since it couples to the second derivative of the flow velocity $\mathbf{u}$, and hence can detect \emph{vorticity gradients} across the flow, this model supports a force term $\mathcal{F}_i$ with nonzero torque i.e. the stress tensor can be non-symmetric \cite{CondiffDahler64, Cosserat09book}, as we will demonstrate in section \textbf{\ref{subsec:closure2}}. If we drop this term, the kinematic situation will reduce to two bead-spring pairs having a fixed common centre, which is the one considered in \cite{Bird87} when considering multibead-chain models.

In more geometrical terms, without the second derivative term, the flow vector field $\mathbf{u}$ will be acting on the configuration space $(T\oplus T)M$ instead of $T^{(2)}M$. Each point in $(T\oplus T)M$ consists of a pair of tangent vectors at a point on $M$. Figures \ref{fig:TplusTM} and \ref{fig:T2M} illustrate the differences graphically. As a consequence, the resulting semidirect product Lie--Poisson structure involving $(T\oplus T)M$ will generally have different kinematic properties to the structure involving $T^{(2)}M$.

This difference is difficult to see if we only consider the configuration space itself but not how it is constructed. For example, when $M=\mathbb{R}^n$, both $(T\oplus T)M$ and $T^{(2)}M$ can be considered as (or more properly, are diffeomorphic to) $\mathbb{R}^{3n}$, but there are nontheless significant differences, namely that for quadratic energy, one model supports torque while the other does not!
(In fact, we can construct an isomorphism of $(T\oplus T)M$ and $T^{(2)}M$ as fibre bundles using a metric -- see appendix \textbf{\ref{subapp:vectTNM}} for more details. However, the isomorphism depends on the metric, and the complete lifts of vector fields to $(T\oplus T)M$ and $T^{(2)}M$ respectively do not coincide under the isomorphism.)

\noindent\makebox[\textwidth][c]{
\begin{minipage}[t]{.48\textwidth}

\begingroup%
  \makeatletter%
  \providecommand\color[2][]{%
    \errmessage{(Inkscape) Color is used for the text in Inkscape, but the package 'color.sty' is not loaded}%
    \renewcommand\color[2][]{}%
  }%
  \providecommand\transparent[1]{%
    \errmessage{(Inkscape) Transparency is used (non-zero) for the text in Inkscape, but the package 'transparent.sty' is not loaded}%
    \renewcommand\transparent[1]{}%
  }%
  \providecommand\rotatebox[2]{#2}%
  \newcommand*\fsize{\dimexpr\f@size pt\relax}%
  \newcommand*\lineheight[1]{\fontsize{\fsize}{#1\fsize}\selectfont}%
  \ifx\svgwidth\undefined%
    \setlength{\unitlength}{290.29578308bp}%
    \ifx\svgscale\undefined%
      \relax%
    \else%
      \setlength{\unitlength}{\unitlength * \real{\svgscale}}%
    \fi%
  \else%
    \setlength{\unitlength}{\svgwidth}%
  \fi%
  \global\let\svgwidth\undefined%
  \global\let\svgscale\undefined%
  \makeatother%
  \begin{picture}(1,0.61114104)%
    \lineheight{1}%
    \setlength\tabcolsep{0pt}%
    \put(0,0){\includegraphics[width=\unitlength]{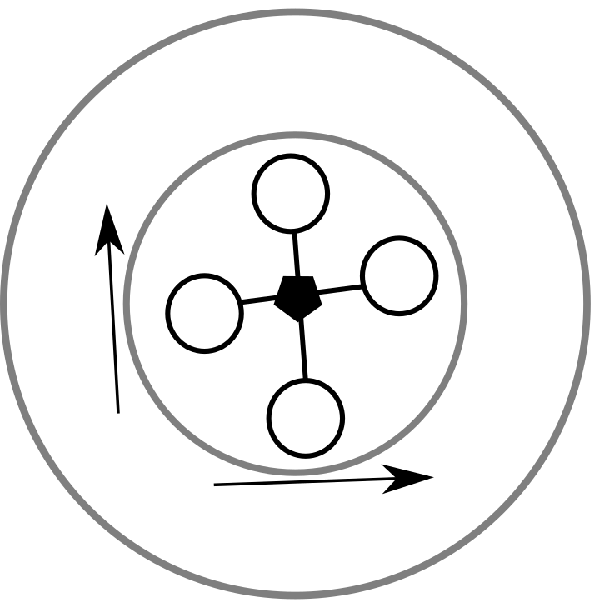}}%
    \put(0.42651937,0.41233539){\color[rgb]{0,0,0}\makebox(0,0)[lt]{\lineheight{1.25}\smash{\begin{tabular}[t]{l}\large$linear$\end{tabular}}}}%
    \put(0.46473652,0.55262468){\color[rgb]{0,0,0}\makebox(0,0)[lt]{\lineheight{1.25}\smash{\begin{tabular}[t]{l}\large$quadratic$\end{tabular}}}}%
    \put(0.08010154,0.27410017){\color[rgb]{0,0,0}\makebox(0,0)[lt]{\lineheight{1.25}\smash{\begin{tabular}[t]{l}\large$y$\end{tabular}}}}%
    \put(0.29392671,0.08300262){\color[rgb]{0,0,0}\makebox(0,0)[lt]{\lineheight{1.25}\smash{\begin{tabular}[t]{l}\large$z$\end{tabular}}}}%
  \end{picture}%
\endgroup%

\captionof{figure}{$(T\oplus T)M$ consists of a point on $M$, and a pair of tangent vectors $y,z$ attached to that point. It can be though of as the configuration space of two bead-spring pairs with a common centre.}\label{fig:TplusTM}

\end{minipage}

\begin{minipage}[t]{.04\textwidth}
\quad
\end{minipage}

\begin{minipage}[t]{.48\textwidth}

\begingroup%
  \makeatletter%
  \providecommand\color[2][]{%
    \errmessage{(Inkscape) Color is used for the text in Inkscape, but the package 'color.sty' is not loaded}%
    \renewcommand\color[2][]{}%
  }%
  \providecommand\transparent[1]{%
    \errmessage{(Inkscape) Transparency is used (non-zero) for the text in Inkscape, but the package 'transparent.sty' is not loaded}%
    \renewcommand\transparent[1]{}%
  }%
  \providecommand\rotatebox[2]{#2}%
  \newcommand*\fsize{\dimexpr\f@size pt\relax}%
  \newcommand*\lineheight[1]{\fontsize{\fsize}{#1\fsize}\selectfont}%
  \ifx\svgwidth\undefined%
    \setlength{\unitlength}{290.29575064bp}%
    \ifx\svgscale\undefined%
      \relax%
    \else%
      \setlength{\unitlength}{\unitlength * \real{\svgscale}}%
    \fi%
  \else%
    \setlength{\unitlength}{\svgwidth}%
  \fi%
  \global\let\svgwidth\undefined%
  \global\let\svgscale\undefined%
  \makeatother%
  \begin{picture}(1,0.61114109)%
    \lineheight{1}%
    \setlength\tabcolsep{0pt}%
    \put(0,0){\includegraphics[width=\unitlength]{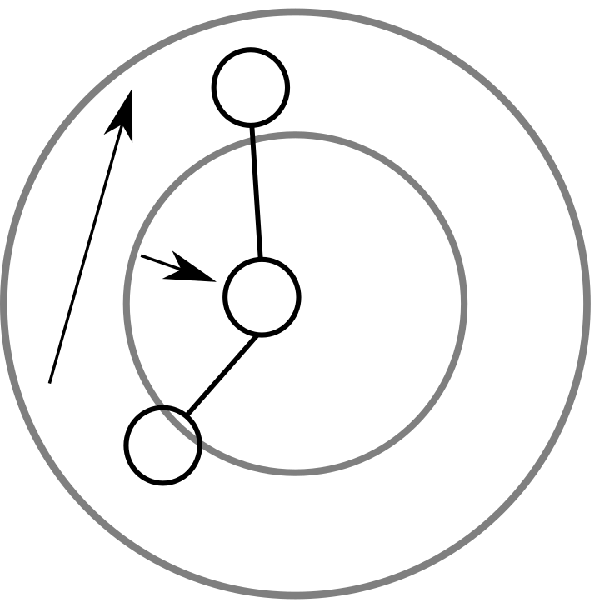}}%
    \put(0.42651938,0.4123354){\color[rgb]{0,0,0}\makebox(0,0)[lt]{\lineheight{1.25}\smash{\begin{tabular}[t]{l}\large$linear$\end{tabular}}}}%
    \put(0.46473645,0.55262472){\color[rgb]{0,0,0}\makebox(0,0)[lt]{\lineheight{1.25}\smash{\begin{tabular}[t]{l}\large$quadratic$\end{tabular}}}}%
    \put(0.05537261,0.34306543){\color[rgb]{0,0,0}\makebox(0,0)[lt]{\lineheight{1.25}\smash{\begin{tabular}[t]{l}\large$y$\end{tabular}}}}%
    \put(0.15042652,0.34511585){\color[rgb]{0,0,0}\makebox(0,0)[lt]{\lineheight{1.25}\smash{\begin{tabular}[t]{l}\large$z$\end{tabular}}}}%
  \end{picture}%
\endgroup%

\captionof{figure}{$T^{(2)}M$ consists of equivalence classes of paths $I\rightarrow M$ with the same Taylor series up to second order, with $y,z$ being the first and second order coefficients in the Taylor expansion respectively. It can be though of as the configuration space of a $3$-bead chain.}\label{fig:T2M}

\end{minipage}

}

\subsection{Dissipation bracket for $3$-bead chains}\label{subsec:diss3bead}

We continue to follow the same strategy as the bead-spring pairs to construct the dissipation bracket for the $3$-bead chain suspension.

Let $M = \mathbb{R}^n$, so we can identify $T^{(2)}M$ with $\mathbb{R}^{3n}$ with coordinates $(x^i,y^i,z^i)$. Consider the usual Riemannian metric $g$ on $\mathbb{R}^{3n}$, which can be written as
\begin{align}
\mathrm{d}s^2 = \delta_{ij}\mathrm{d}x^i\mathrm{d}x^j + \delta_{ij}\mathrm{d}y^i\mathrm{d}y^j + \delta_{ij}\mathrm{d}z^i\mathrm{d}z^j .
\end{align}
The construction of Riemannian metrics on higher order tangent bundles of general Riemannian manifolds is considered in appendix \textbf{\ref{app:vect}}. 

We can define the dissipation bracket
\begin{align}
(F,G) = \int \mathrm{d}^nx\mathrm{d}^ny\mathrm{d}^nz \ \psi \ \widetilde{g}\left(  \mathrm{d}\left(\frac{\delta F}{\delta \psi}\right), \mathrm{d}\left(\frac{\delta G}{\delta \psi}\right) \right),
\end{align}
where $\widetilde{g}$ denotes the inverse of $g$, so the dynamics is given by
\begin{align}
\dt{F} = \{F,H\} - \frac{1}{\zeta}(F,H),
\end{align}
for some positive parameter $1/\zeta$ called the \emph{mobility}. This dissipation bracket provides an implementation of a linear mobility relation between the applied force and the relative velocity to the surrounding fluid.

In coordinates, this dissipation bracket is
\begin{align}
(F,G) =& \int \mathrm{d}^nx\mathrm{d}^ny\mathrm{d}^nz \ \psi  \Bigg[ \delta^{ij}\frac{\partial}{\partial x^i}\left( \frac{\delta F}{\delta \psi} \right)\frac{\partial}{\partial x^j}\left( \frac{\delta F}{\delta \psi} \right) \nonumber \\
& \qquad + \delta^{ij}\frac{\partial}{\partial y^i}\left( \frac{\delta F}{\delta \psi} \right)\frac{\partial}{\partial y^j}\left( \frac{\delta F}{\delta \psi} \right) + \delta^{ij}\frac{\partial}{\partial z^i}\left( \frac{\delta F}{\delta \psi} \right)\frac{\partial}{\partial z^j}\left( \frac{\delta F}{\delta \psi} \right)\Bigg] .
\end{align}
With the dissipation bracket, the Hamiltonian should now be interpreted as a free energy. As a generalisation to (\ref{freeE1}), consider free energies of the form
\begin{align}
H_s[\psi] = \int \mathrm{d}^nx\mathrm{d}^ny\mathrm{d}^nz \ \left(E(y,z)\psi +k_BT\psi \log \psi\right),
\end{align}
with variational derivative ${\delta H_s}/{\delta \psi} = E(y,z) + k_B T \left(\log\psi + 1\right)$. The force can be written as the divergence of a stress tensor, $\mathcal{F}_i = {\partial \sigma^j_{i}}/{\partial x^j}$ in (\ref{force2}), where
\begin{align}
\label{stress2}
\sigma^j_{i} = -3k_B T \int \mathrm{d}^ny\mathrm{d}^nz \ \psi + \int \mathrm{d}^ny\mathrm{d}^nz \ \psi \left(y^j\frac{\partial E}{\partial y^j} + z^j\frac{\partial E}{\partial z^j}\right) - \frac{\partial}{\partial x^k} \int \mathrm{d}^ny\mathrm{d}^nz \ \psi y^jy^k \frac{\partial E}{\partial z^k},
\end{align}
where we have used integration by parts.

The evolution equation for $\psi$ is obtained similarly:
\begin{align} \label{psieom2dis}
\dt{\psi} &+ \frac{\partial}{\partial x^i}\left(u^i \psi \right) + \frac{\partial}{\partial y^i}\left(\frac{\partial u^i}{\partial x^j}y^j \psi\right) + \frac{\partial}{\partial z^i}\left(\left[ \frac{\partial^2 u^i}{\partial x^j \partial x^k}y^j y^k + \frac{\partial u^i}{\partial x^j}z^j \right]\psi \right) \nonumber \\
&= \frac{1}{\zeta}\frac{\partial}{\partial y^i}\left(\delta^{ij}\frac{\partial E}{\partial y^j}\right) + \frac{1}{\zeta}\frac{\partial}{\partial z^i}\left(\delta^{ij}\frac{\partial E}{\partial z^j}\right) + \frac{k_B T}{\zeta}\left(  \nabla^2_x\psi + \nabla^2_y\psi + \nabla^2_z\psi\right),
\end{align}
where $\nabla^2_x = \delta^{ij}{\partial^2}/{\partial x^i \partial x^j}$, $\nabla^2_y = \delta^{ij}{\partial^2}/{\partial y^i \partial y^j}$, and $\nabla^2_z = \delta^{ij}{\partial^2}/{\partial z^i \partial z^j}$.

For general forms of $E(y,z)$, there is no known exact method to express the evolution of the relevant moments appearing in the stress tensor in terms of a finite set of moments in $\psi$, i.e. there is no finite closure in general. However, a finite closure exists for a quadratic energy of the form
\begin{align}\label{E2}
E(y,z) = \frac{\kappa_1}{2}\delta_{ij}y^iy^j + \frac{\kappa_2}{2}\delta_{ij}z^iz^j.
\end{align}
By writing $\sigma^{jk} = \sigma^{j}_i\delta^{ik}$ and using the notation $\left \langle \cdots \right\rangle = \int \mathrm{d}^ny\mathrm{d}^nz \ \psi \left(\cdots\right)$, we can write the stress tensor (\ref{stress2}) as
\begin{align} \label{stress2final}
\sigma^{jk} = -3k_BT\left \langle 1 \right\rangle\delta^{jk} + \kappa_1\left \langle y^j y^k \right\rangle + \kappa_2\left \langle z^j z^k \right\rangle - \kappa_2 \frac{\partial}{\partial x^l} \left \langle y^jy^lz^k \right\rangle .
\end{align}
The evolution equation (\ref{psieom2dis}) for $\psi$ then simplifies to
\begin{align} \label{psieom2final}
\dt{\psi} &+ \frac{\partial}{\partial x^i}\left(u^i \psi \right) + \frac{\partial}{\partial y^i}\left(\frac{\partial u^i}{\partial x^j}y^j \psi\right) + \frac{\partial}{\partial z^i}\left(\left[ \frac{\partial^2 u^i}{\partial x^j \partial x^k}y^j y^k + \frac{\partial u^i}{\partial x^j}z^j \right]\psi \right) \nonumber \\
&= \frac{1}{\zeta}\frac{\partial}{\partial y^i}\left(y^i \psi\right) + \frac{1}{\zeta}\frac{\partial}{\partial z^i}\left(z^i \psi\right) + \frac{k_B T}{\zeta}\left(  \nabla^2_x\psi + \nabla^2_y\psi + \nabla^2_z\psi\right) .
\end{align}
Note that
\begin{enumerate}
\item The stress tensor $\sigma^{jk}$ in (\ref{stress2final}) is not manifestly symmetric, but the system nonetheless conserves total fluid angular momentum. The torque of the suspension on the fluid is the antisymmetrised stress tensor: $\tau^{jk} = \sigma^{jk} - \sigma^{kj}$. For this specific form of $\sigma^{jk}$, the torque $\tau^{jk}$ is a divergence of a $3$-index tensor. This means the torque terms are not sources or sinks of angular momentum, but \emph{angular momentum fluxes} \cite{CondiffDahler64}. These extra angular momentum fluxes can be thought of as the transmission of angular momentum between adjacent fluid parcels across a common material surface, through the bending of chains that cross the surface. Since the chain has no inertia, the torques on the chain must balance at all times, which means the transmission of angular momentum is instantaneous.

There is a slightly different description of this asymmetric stress tensor in terms of generalised continuum systems that can possibly have internal angular momentum, such as polar fluids \cite{CondiffDahler64, Rosensweig85}. These generalised continuum systems were first considered by Cosserat and Cosserat \cite{Cosserat09book}. In such a system, the rate of change of total angular momentum in a Lagrangian control volume consists of sources such as body forces and body torques, as well as fluxes through the boundary of the control volume. There are two types of angular momentum fluxes, the first being the \emph{hydrodynamic angular momentum flux} $-\sigma \times \mathbf{x}$, which is the angular momentum flux generated by the fluid stress, and the second being a \emph{couple stress}, an angular momentum flux generated by the interaction of the internal degrees of freedom. The couple stress depends on the intrinsic properties of the fluid, in the same way the hydrodynamic stress depends on the thermodynamic equation of state, viscosity etc. of the fluid.
The hydrodynamic torque $\mathbf{x} \times \nabla \cdot \sigma$ can be separated into two parts
\begin{align}
\mathbf{x} \times \nabla \cdot \sigma = - \nabla \cdot \left( \sigma \times \mathbf{x} \right) - \tau,
\end{align}
where the first term $- \nabla \cdot \left( \sigma \times \mathbf{x} \right)$ is the divergence of the hydrodynamic angular momentum flux, and the second term $\tau$ is (up to conventions on the sign and factors of $2$) the antisymmetric, or pseudovector part of the hydrodynamic stress tensor $\sigma$. The hydrodynamic torque is not necessarily the divergence of the hydrodynamic angular momentum flux, and the \enquote{excess} term $\tau$ represents the exchange of angular momentum between the internal and fluid degrees of freedom.

In our system of Hookean $3$-bead chains, there are no body forces or body torques. Moreover, since the $3$-bead chains are immersed in a Stokes flow, they have no inertia and hence no internal angular momentum. Therefore, there must be an instantaneous torque balance -- the  \enquote{excess} term $\tau$, which is the antisymmetric part of the hydrodynamic stress tensor $\sigma$, must be balanced by the divergence of a couple stress. In our scenario, $\tau$ itself is also an exact divergence, so this is possible, and we can identify $\tau$ with (the negative of) the couple stress itself. This is to be contrasted with systems with internal \enquote{spin} degrees of freedom such as ferrofluids, where the couple stress alone cannot balance the \enquote{excess} term in (or antisymmetric part of) the hydrodynamic stress \cite{Shliomis74, Rosensweig85}.

The fact the system without the dissipation terms conserves linear and angular momentum is consistent with \emph{Noether's theorem}, since the energy $E(y,z)$ in (\ref{E2}) is translationally and rotationally invariant. Moreover, since the linear and angular momentum do not depend on $\psi$, they are not affected by the dissipation bracket, so the double bracket system as a whole conserves linear and angular momentum.

\item If we apply $\int \mathrm{d}^n z$ to the evolution equation (\ref{psieom2final}) of $\psi$, we will obtain the previous evolution equation (\ref{psi1eomfinal}) for the bead-spring pairs (up to numerical factors in the dissipation bracket). This is analogous to \enquote{forgetting} the middle bead.
\end{enumerate}

As we shall see in the next section, the moments appearing in the stress tensor (\ref{stress2final}) do not quite form a closed system, but this can be fixed by including a few extra moments.

\subsection{Closure of the 3-bead chain model with quadratic energy}\label{subsec:closure2}
By a (finite and exact) \emph{closure} we mean that there is a finite set of moments $\mu_1,\cdots,\mu_I$ of $\psi$, i.e. quantities of the form $\left\langle p(y,z)\right\rangle$ for a polynomial $p$, such that the evolution of each moment can be expressed using this set of moments, together with $\mathbf{u}$ and its spatial derivatives:
\begin{align}
\dt{\mu}_J = f_J\left( \mu_1,\cdots,\mu_I; \mathbf{u} \right) \quad \text{for $J = 1,\ldots,I$},
\end{align}
and such that $\mathbf{\sigma} = \mathbf{\sigma}\left( \mu_1,\cdots,\mu_I; \mathbf{u} \right) $ i.e. the stress tensor is completely described by these moments. 

We will explicitly show that the 3-bead chain model with quadratic energy, described by the equations (\ref{stress2final}, \ref{psieom2final},) has such a closure. Since $\left\langle 1 \right\rangle, \left\langle y^jy^k \right\rangle, \left\langle z^jz^k \right\rangle, \left\langle y^iy^lz^k \right\rangle$ all appear in the expression for $\sigma^{jk}$, let us look at their time derivatives:
\begin{align}
\dt{\left\langle 1 \right\rangle} &+ \frac{\partial}{\partial x^i}\left(u^i \left\langle 1 \right\rangle\right) = \frac{k_BT}{\zeta}\nabla^2_x \left\langle 1 \right\rangle, \\
\dt{\left\langle y^j y^k \right\rangle} &+ \frac{\partial}{\partial x^i}\left( u^i \left\langle y^j y^k \right\rangle\right) - \frac{\partial u^j}{\partial x^l}\left\langle y^l y^k \right\rangle - \frac{\partial u^k}{\partial x^l}\left\langle y^j y^l \right\rangle \nonumber \\
=& -\frac{2\kappa_1}{\zeta}\left\langle y^j y^k \right\rangle  + \frac{2k_BT}{\zeta}\delta^{jk}\left\langle 1 \right\rangle + \frac{k_BT}{\zeta}\nabla^2_x \left\langle y^j y^k \right\rangle, \\
\dt{\left\langle z^j z^k \right\rangle} &+ \frac{\partial}{\partial x^i}\left( u^i \left\langle z^j z^k \right\rangle\right) - \frac{\partial u^j}{\partial x^l}\left\langle z^l z^k \right\rangle - \frac{\partial u^k}{\partial x^l}\left\langle z^j z^l \right\rangle - \frac{\partial^2 u^j}{\partial x^m \partial x^l} \left\langle y^my^l z^k\right\rangle -  \frac{\partial^2 u^k}{\partial x^m \partial x^l} \left\langle y^my^l z^j\right\rangle \nonumber \\
=& -\frac{2\kappa_2}{\zeta}\left\langle z^j z^k \right\rangle  + \frac{2k_BT}{\zeta}\delta^{jk}\left\langle 1 \right\rangle + \frac{k_BT}{\zeta}\nabla^2_x \left\langle z^j z^k \right\rangle, \\
\dt{\left\langle y^jy^l z^k\right\rangle} &+ \frac{\partial}{\partial x^i}\left( u^i \left\langle y^jy^l z^k\right\rangle \right) - \frac{\partial u^j}{\partial x^m}\left\langle y^my^l z^k\right\rangle - \frac{\partial u^l}{\partial x^m}\left\langle y^jy^m z^k\right\rangle - \frac{\partial u^k}{\partial x^m}\left\langle y^jy^l z^m\right\rangle - \frac{\partial^2 u^k}{\partial x^m \partial x^n} \left\langle y^my^ny^jy^l\right\rangle \nonumber \\
=& -\frac{2\kappa_1}{\zeta}\left\langle y^j y^l z^k \right\rangle  - \frac{\kappa_2}{\zeta}\left\langle y^j y^l z^k \right\rangle + \frac{2k_BT}{\zeta}\delta^{jl}\left\langle z^k \right\rangle + \frac{k_BT}{\zeta}\nabla^2_x \left\langle y^j y^l z^k \right\rangle.
\end{align}
So the moments $\left\langle 1 \right\rangle, \left\langle y^jy^k \right\rangle, \left\langle z^jz^k \right\rangle, \left\langle y^iy^lz^k \right\rangle$ do not quite form a closed system, as their time evolution depends on the extra moments $\left\langle y^my^ny^jy^l\right\rangle$ and $\left\langle z^k \right\rangle$. However, we can close the system by including the time evolution of just these extra moments:
\begin{align}
\dt{\left\langle y^my^ny^jy^l\right\rangle} &+ \frac{\partial}{\partial x^i}\left( u^i \left\langle y^my^ny^jy^l\right\rangle \right) - \frac{\partial u^m}{\partial x^i}\left\langle y^iy^ny^jy^l\right\rangle - \frac{\partial u^n}{\partial x^i}\left\langle y^my^iy^jy^l\right\rangle - \frac{\partial u^j}{\partial x^i}\left\langle y^my^ny^iy^l\right\rangle - \frac{\partial u^l}{\partial x^i}\left\langle y^my^ny^jy^i\right\rangle \nonumber \\
=& -\frac{4\kappa_1}{\zeta}\left\langle y^m y^n y^j y^l \right\rangle  + \frac{k_BT}{\zeta}\nabla^2_x \left\langle y^m y^n y^j y^l \right\rangle \nonumber \\
& \qquad + \frac{2k_BT}{\zeta} \left( \delta^{mn} \left\langle y^j y^l \right\rangle +  \delta^{mj} \left\langle y^n y^l \right\rangle +  \delta^{ml} \left\langle y^n y^j \right\rangle +  \delta^{nj} \left\langle y^m y^l\right\rangle +  \delta^{nl} \left\langle y^m y^j \right\rangle + \delta^{jl} \left\langle y^m y^n \right\rangle  \right), \\
\dt{\left\langle z^k \right\rangle} &+ \frac{\partial}{\partial x^i}\left(u^i \left\langle z^k \right\rangle\right) - \frac{\partial u^k}{\partial x^i}\left\langle z^i \right\rangle -  \frac{\partial^2 u^k}{\partial x^m \partial x^n} \left\langle y^my^n\right\rangle = -\frac{\kappa_2}{\zeta}\left\langle z^k \right\rangle + \frac{k_BT}{\zeta}\nabla_x^2 \left\langle z^k \right\rangle.
\end{align}
Thus the moments $\left\langle 1 \right\rangle, \left\langle y^jy^k \right\rangle, \left\langle z^jz^k \right\rangle, \left\langle y^iy^lz^k \right\rangle, \left\langle y^my^ny^jy^l\right\rangle, \left\langle z^k \right\rangle$ form a closed system of evolution equations given the flow field $\mathbf{u}$, and are also sufficient to describe the stress tensor $\sigma^{jk}$ in (\ref{stress2final}) completely. It is possible to interpret these moments as tensor fields -- see appendix \textbf{\ref{app:vect}} for more details.

This is potentially a viable model for describing suspensions of molecules with stretching and bending degrees of freedom. This type of suspension can transmit angular momentum between adjacent fluid parcels.

\section{Explicit formulae for the multibead-chain model}\label{sec:TNMmodelcont}

Having investigated the $3$-bead chain as a concrete example, we now return to the general case. In section \textbf{\ref{sec:TNMmodel}} we have considered the action of vector fields $\mathbf{u}\in\mathrm{Vect}(M)$ on functions on the $N^{th}$ order tangent bundle $T^{(N)}M$, which can be thought of as the configuration space of small $(N+1)$-bead chains living on a manifold $M$. We also considered the semidirect product relevant to the advection of such multibead chains in an ideal compressible fluid, and obtained the following semidirect Lie--Poisson bracket.

If $(\mathbf{m},\rho,\psi)$ denote the fluid momentum density, fluid mass density, and the distribution function of multibead-chains in configuration space, respectively, then for functionals $F,G$, the Poisson bracket is
\begin{align}
\{F,G\}[\mathbf{m},\rho,\psi] =& -\left\langle \mathbf{m},\left[\frac{\delta F}{\delta \mathbf{m}},\frac{\delta G}{\delta \mathbf{m}}\right] \right\rangle - \left\langle \rho, \frac{\delta F}{\delta \mathbf{m}} \cdot \frac{\delta G}{\delta \rho} - \frac{\delta G}{\delta \mathbf{m}} \cdot \frac{\delta F}{\delta \rho} \right\rangle - \left\langle \psi, \left(\frac{\delta F}{\delta \mathbf{m}}\right)^{\#} \cdot \frac{\delta G}{\delta \psi} - \left(\frac{\delta G}{\delta \mathbf{m}}\right)^{\#} \cdot \frac{\delta F}{\delta \psi} \right\rangle, \nonumber \\
    =& \{F,G\}_{fluids} + \{F,G\}_{\psi},
\end{align}
where $\{F,G\}_{fluids}$ denotes the usual compressible fluid bracket, and the subbracket $\{F,G\}_{\psi}$ denotes the terms that explicitly involve $\psi$.

In the following we work exclusively in coordinates. As usual, let $x^i$ be coordinates on $M$, and $(x^i,y^i_{(1)},\ldots,y^i_{(N)})$ be the induced coordinates on $T^{(N)}M$. For brevity, we denote the standard volume element in the space of the \enquote{internal degrees of freedom} with respect to this coordinate system as
\begin{align}
\mathrm{d}\Gamma = \mathrm{d}^n y_{(1)} \cdots \mathrm{d}^n y_{(N)} .
\end{align}
To write the components of the complete lift $\mathbf{u}^{\#}$ in terms of the components of $\mathbf{u}(x)$ and its derivatives, we define the \emph{exponential} of an operator $X$ by the formal power series $\exp(X) = 1+X+X^2/2! + \ldots$, and write
\begin{align}
u^i(x + \delta x) = \exp\left( \delta x^j \frac{\partial}{\partial x^j} \right) u^i(x) = \sum_{n=0}^{\infty} \frac{1}{n!}\left( \delta x^j \frac{\partial}{\partial x^j} \right)^n u^i(x).
\end{align}
In the following we will only need a finite number of terms in the series, so we will never encounter situations where the convergence of the series is an issue.

The components of $\mathbf{u}^{\#}$ in the coordinate system $(x^i,y^i_{(1)},\ldots,y^i_{(N)})$ can be written as $(u^i,u^i_{(1)}\ldots,u^i_{(N)})$, where
\begin{align}\label{uia}
u^i_{(a)} &= \left(\frac{\partial}{\partial t}\right)^a\bigg\vert_{t=0} u^i \left( x + \sum_{b=1}^{N} \frac{t^b}{b!}y_{(b)} \right), \nonumber \\ 
&= \left(\frac{\partial}{\partial t}\right)^a\left[ \exp\left( \sum_{b=1}^{N} \frac{t^b}{b!}y^j_{(b)} \frac{\partial}{\partial x^j} \right) \right]_{t=0} \cdot u^i(x), \nonumber \\
&= \mathcal{P}_a\left( y_{(1)}, \ldots, y_{(N)}, \frac{\partial}{\partial x} \right) \cdot u^i(x) .
\end{align}
The expressions $\mathcal{P}_a$ are defined to be the $a^{th}$ $t$-derivative of the exponential operator in the second line of (\ref{uia}). We can equivalently characterise $\mathcal{P}_a$ in terms of a generating function 
\begin{align}\label{defPa}
\exp\left( \sum_{b=1}^{N} \frac{t^b}{b!}y^j_{(b)} \frac{\partial}{\partial x^j} \right) = 1 + \sum_{a=1}^N \frac{t^a}{a!} \mathcal{P}_a\left( y_{(1)}, \ldots, y_{(N)}, \frac{\partial}{\partial x} \right) + O(t^{N+1}).
\end{align}
By collecting powers of $t$ in the exponential series, it can be directly seen that $\mathcal{P}_a$ is a polynomial in $y_{(1)}, \ldots, y_{(N)}, {\partial}/{\partial x}$. The variables commute because ${\partial}/{\partial x}$ commutes with $y_{(1)}, \ldots, y_{(N)}$.
For brevity, we will write $\mathcal{P}_a = \mathcal{P}_a( y, \partial/\partial x)$ to indicate its dependence on the $y$-coordinates, and emphasise that it is a differential operator acting on quantities with $x$-dependence.

We can verify a few crucial properties of $\mathcal{P}_a$ by collecting powers of $t$ in (\ref{defPa}):
\begin{enumerate}[(i)]
\item $a$ is the highest power of ${\partial}/{\partial x}$ appearing $\mathcal{P}_a$. 
\item $1$ is the lowest power of ${\partial}/{\partial x}$ appearing $\mathcal{P}_a$. This means there is no \enquote{constant term} in $\mathcal{P}_a$ as a polynomial in ${\partial}/{\partial x}$, which is a fact we use later to express the force as the divergence of a stress tensor for some specific forms of the Hamiltonian functional.
\item If we define the order of each monomial $y_{(b)}$ to be $b$, then in this sense $\mathcal{P}_a$ is a polynomial of order $a$ in the variables $y_{(1)}, \ldots, y_{(N)}$ with coefficients in ${\partial}/{\partial x}$. More concretely, if we perform the scalings $y_{(b)} \mapsto \lambda^b y_{(b)}$ for some real $\lambda$, then $\mathcal{P}_a \mapsto \lambda^a \mathcal{P}_a$. In particular, $\mathcal{P}_a$ does not depend on $y_{(b)}$ for $b>a$. We will return to this property when we consider the problem of exact closure.
\end{enumerate}

We can thus write the $\psi$-subbracket of functionals $F, G$ on $(\mathbf{m},\rho,\psi)$ as
\begin{align}
\{F,G\}_{\psi}  =& - \int \mathrm{d}^n x\mathrm{d}\Gamma \ \psi \left[ \frac{\delta F}{\delta m_i}\frac{\partial}{\partial x^i}\left( \frac{\delta G}{\delta \psi} \right) - \frac{\delta G}{\delta m_i}\frac{\partial}{\partial x^i}\left( \frac{\delta F}{\delta \psi} \right)\right] \nonumber \\
& \qquad - \int \mathrm{d}^n x\mathrm{d}\Gamma \ \psi \sum_{a=1}^N \left(\mathcal{P}_a\left(y,\frac{\partial}{\partial x}\right) \cdot \frac{\delta F}{\delta m_i} \right) \left(\frac{\partial}{\partial y^i_{(a)}} \frac{\delta G}{\delta \psi} \right) \nonumber \\
& \qquad + \int \mathrm{d}^n x\mathrm{d}\Gamma \ \psi \sum_{a=1}^N \left(\mathcal{P}_a\left(y,\frac{\partial}{\partial x}\right) \cdot \frac{\delta G}{\delta m_i} \right) \left(\frac{\partial}{\partial y^i_{(a)}} \frac{\delta F}{\delta \psi} \right), \nonumber \\
=&  \int \mathrm{d}^n x\mathrm{d}\Gamma \ \frac{\delta F}{\delta m_i} \left[ -\psi\frac{\partial}{\partial x^i}\left( \frac{\delta G}{\delta \psi} \right) - \sum_{a=1}^N \mathcal{P}_a\left(y,-\frac{\partial}{\partial x}\right) \cdot \left(  \psi \frac{\partial}{\partial y^i_{(a)}}\left( \frac{\delta G}{\delta \psi} \right)\right)  \right] \nonumber \\
& \qquad + \int \mathrm{d}^n x\mathrm{d}\Gamma \ \frac{\delta F}{\delta \psi} \left[ -\frac{\partial}{\partial x^i}\left( \frac{\delta G}{\delta m_i} \psi \right) - \sum_{a=1}^N \frac{\partial}{\partial y^i_{(a)}}\left( \left( \mathcal{P}_a\left(y,\frac{\partial}{\partial x}\right) \cdot \frac{\delta G}{\delta m_i} \right) \psi \right) \right].
\end{align}
We have repeatedly integrated by parts to obtain the terms multiplying to $\delta F / \delta m_i$ and $\delta F / \delta \psi$, using the fact that $\mathcal{P}_a$ is a polynomial of degree $a$ in $\partial /\partial x$. Given a Hamiltonian functional $H$, we can now read off the extra body force on the fluid $\mathcal{F}_i$ and the time evolution of the distribution function $\psi$ as
\begin{align}\label{forceN}
\mathcal{F}_i & =  \int \mathrm{d}\Gamma \ \left[ -\psi\frac{\partial}{\partial x^i}\left( \frac{\delta H}{\delta \psi} \right) - \sum_{a=1}^N \mathcal{P}_a\left(y,-\frac{\partial}{\partial x}\right) \cdot \left(  \psi \frac{\partial}{\partial y^i_{(a)}}\left( \frac{\delta H}{\delta \psi} \right)\right)  \right], \\
\label{psiN}
\dt{\psi} & = -\frac{\partial}{\partial x^i}\left( \frac{\delta H}{\delta m_i} \psi \right) - \sum_{a=1}^N \frac{\partial}{\partial y^i_{(a)}}\left( \psi  \mathcal{P}_a\left(y,\frac{\partial}{\partial x}\right) \cdot \frac{\delta H}{\delta m_i} \right).
\end{align}

\subsection{Conservation of fluid momentum for the multibead-chain models}\label{subsec:momentumN}

We now specialise to Hamiltonians of the form
\begin{align}\label{energyN}
H[\mathbf{m},\rho,\psi] = H_{fluids}[\mathbf{m},\rho] + H_s[\psi], \quad H_s[\psi] = \int \mathrm{d}^nx \mathrm{d} \Gamma \left( E(y)\psi + k_B T \psi\log\psi \right),
\end{align}
where $E(y)$ is some function of the variables $y_{(1)}, \ldots, y_{(N)}$. The force on the fluid is given by (\ref{forceN}) as
\begin{align}
\mathcal{F}_i = - (N+1) k_B T \frac{\partial}{\partial x^i}\left( \int \mathrm{d}\Gamma \ \psi \right) - \int \mathrm{d}\Gamma \ \sum_{a=1}^N \mathcal{P}_a\left(y,-\frac{\partial}{\partial x}\right) \cdot \left(  \psi \frac{\partial E}{\partial y^i_{(a)}}\right).
\end{align}
For each $a$, the polynomial $\mathcal{P}_a( y, -\partial/\partial x)$ can be written as 
\begin{align}
\mathcal{P}_a\left( y, -\frac{\partial}{\partial x}\right) = \mathcal{Q}_a^j  \left( y, -\frac{\partial}{\partial x} \right) \left( - \frac{\partial}{\partial x^j}\right) = - y_{(a)}^j \frac{\partial}{\partial x^j} + O\left( \frac{\partial^2}{\partial x^2} \right),
\end{align}
where $\mathcal{Q}_a^j  ( y, -\partial/\partial x)$ is some other polynomial. This can be directly verified from the definition of $\mathcal{P}_a( y, -\partial/\partial x)$, as in (\ref{defPa}), using property (ii). The important property here is that $\mathcal{Q}_a^j  ( y, -\partial/\partial x) = y^j_{(a)} + O(\partial/\partial x)$, meaning that the only term in the operator $\mathcal{Q}_a^j( y, -\partial/\partial x)$ that does not take $x$-derivatives of its argument is $y^j_{(a)}$.

We can thus write the force as the divergence of a stress tensor, $\mathcal{F}_i = {\partial \sigma^j_{i}}/{\partial x^j}$, for
\begin{align}\label{stressN1}
\sigma^j_{i} = - (N+1) k_B T \delta^j_{i} \int \mathrm{d}\Gamma \ \psi + \int \mathrm{d}\Gamma \ \sum_{a=1}^N \mathcal{Q}^j_a\left(y,-\frac{\partial}{\partial x}\right) \cdot \left(  \psi \frac{\partial E}{\partial y^i_{(a)}}\right).
\end{align}
In particular, if $E(y)$ is a \emph{polynomial} in $y_{(1)}, \cdots, y_{(N)}$ , then $\sigma^j_{i}$ depends on a finite number of \emph{polynomial moments} of $\psi$, and their $x$-derivatives.

The fact that the force can be written as the divergence of a stress tensor implies the conservation of fluid linear momentum. For certain forms of the energy per chain $E(y)$, we can show that the antisymmetric part of the stress tensor can be written as the divergence of a $3$-index tensor, a condition that is sufficient to guarantee conservation of fluid angular momentum.

For example, suppose $E(y)$ depends on the variables $y^i_{(a)}$ only through their \emph{squares} $s_{(a)} = ({1}/{2})\delta_{ij}y^i_{(a)}y^j_{(a)}$, so that
\begin{align}
E(y) = E\left(s_{(1)},\ldots,s_{(N)}\right).
\end{align}
As ${\partial E}/{\partial y^i_{(a)}} = \delta_{ik}y^k_{(a)}{\partial E}/{\partial s_{(a)}}$, the stress tensor $\sigma^{jk} = \sigma^j_{i}\delta^{ik}$ can be written as
\begin{align}\label{stressN2}
\sigma^{jk} = - (N+1) k_B T \delta^{jk}\left( \int \mathrm{d}\Gamma \ \psi \right) + \sum_{a=1}^N\left(\int \mathrm{d}\Gamma \ y^j_{(a)}y^k_{(a)}\frac{\partial E}{\partial s_{(a)}}\psi\right) + \frac{\partial}{\partial x^l} T^{jlk},
\end{align}
where the tensor $T^{jlk}$ depends on $\psi$ and the explicit form of $E$. This automatically guarantees that the antisymmetric part of the stress tensor, which is the torque, is a divergence, i.e. can be considered as an angular momentum flux $A^{jlk} = T^{jlk}-T^{klj}$. In terms of generalised continuum systems, this term has a similar interpretation to the asymmetric stress for the $3$-bead chain suspensions considered in section \textbf{\ref{subsec:diss3bead}}, as a term that is balanced by the divergence of a couple stress.

These conservation laws correspond to the isometries of Euclidean space, by Noether's theorem, and on general Riemannian manifolds they have to be replaced with the isometries of the corresponding metric.

\subsection{Closure of the multibead-chain models without dissipation}\label{subsec:closureN1}

Having investigated the force term in (\ref{forceN}) and the relationship between the energy and the conservation laws in section \textbf{\ref{subsec:momentumN}}, we proceed to investigate the evolution equation (\ref{psiN}) for $\psi$ in the multibead-chain model without dissipation. If $E(y)$ is a polynomial in $y_{(1)}, \cdots, y_{(N)}$, then as we have seen in the last section, the stress tensor depends on a finite number of polynomial moments of $\psi$. We will show that, in this system, we can achieve closure with a finite set of polynomial moments.

We will start by defining the \emph{order} of a monomial in $y_{(1)}, \ldots, y_{(N)}$ to be the sum of its subscripts. For example, $y^i_{(1)}y^j_{(3)}$ has order $4$, while $y^i_{(5)}y^j_{(5)}y^k_{(5)}$ has order $15$. A simple way to keep track of the order is to attach a factor of $\lambda^a$ to each $y_{(a)}$, for $a = 1,\ldots,N$. The \emph{order} of a polynomial $p(y)$ in $y_{(1)}, \ldots, y_{(N)}$ is defined as the maximum order of its monomial terms. We also say a polynomial $p(y)$ is \emph{homogeneous of order $b$} if the scaling $y_{(a)}\mapsto \lambda^a y_{(a)}$ sends $p(y) \mapsto \lambda^b p(y)$. Property (iii) of $\mathcal{P}_a( y, \partial/\partial x)$ derived from (\ref{defPa}) means that $\mathcal{P}_{a}$ is homogeneous with order $a$ as a polynomial in $y$ with coefficients in ${\partial}/{\partial x}$.

Given a polynomial $p(y)$ which is homogeneous of order $b$, consider the time evolution of the corresponding polynomial moment $\left\langle p(y) \right\rangle = \int \mathrm{d}\Gamma \ p(y) \psi$. We will use angle brackets to denote integration over all $y$-space (or internal degrees of freedom) against $\psi$. This is explicitly given by
\begin{align}
\dt{\left\langle p(y) \right\rangle} &= - \frac{\partial}{\partial x^i}\left( u^i(x)  \left\langle p(y) \right\rangle \right) - \sum_{a=1}^N \int \mathrm{d}\Gamma \  p(y) \frac{\partial}{\partial y^i_{(a)}} \left( \psi \mathcal{P}_a\left(y,\frac{\partial}{\partial x}\right) \cdot u^i(x) \right), \nonumber \\
& = - \frac{\partial}{\partial x^i}\left( u^i(x)  \left\langle p(y) \right\rangle \right) + \sum_{a=1}^N \left\langle  \frac{\partial p(y)}{\partial y^i_{(a)}} \mathcal{P}_a\left(y,\frac{\partial}{\partial x}\right)  \right\rangle \cdot u^i(x),
\end{align}
where the factors of ${\partial}/{\partial x}$ in $\mathcal{P}_a$ act on $u^i$ (and not on $\psi$). If $p(y)$ is homogeneous with order $b$, then ${\partial p(y)}/{\partial y^i_{(a)}}$ is either identically zero, or a homogeneous polynomial of order $b-a$. Therefore
\begin{align*}
\frac{\partial p(y)}{\partial y^i_{(a)}} \mathcal{P}_a\left(y,\frac{\partial}{\partial x}\right)\cdot u^i(x)
\end{align*}
is a homogeneous polynomial in $y$ with order $b$, with coefficients in $u^i(x)$ and its spatial derivatives.

This immediately implies that, given the flow field $u^i(x)$, the collection of homogeneous polynomials of order $b$ form a closed system, for each $b$. If in addition, $E(y)$ is a polynomial in $y_{(1)}, \ldots, y_{(N)}$, then the stress tensor will depend on a finite number of polynomial moments of $\psi$ (and their $x$-derivatives). If the highest order of the polynomials that appear in the expression of the stress tensor is $b_{max}$, we can achieve a finite closure by simply collect up the evolution equations for all monomial moments with order at most $b_{max}$.

Note that closure can still be achieved if the time evolution of higher order polynomial moments depends on both equal and lower order polynomial moments; but if the time evolution of lower order polynomial moments depend on higher order polynomial moments, then an exact, finite closure is impossible in general. We will return to this point when we consider the multibead-chain model with dissipation.

\subsection{The multibead-chain models with dissipation}\label{subsec:closureN2}

Now we consider the multibead-chain model with a dissipation bracket. Again, consider the Riemannian metric $g$ on the configuration space of the multibead chain, with line element written as
\begin{align}
\mathrm{d}s^2 = \delta_{ij} \mathrm{d}x^i\mathrm{d}x^j + \sum_{a=1}^N \delta_{ij} \mathrm{d}y_{(a)}^i\mathrm{d}y_{(a)}^j ,
\end{align}
and let $\widetilde{g}$ denote the inverse of $g$. (See appendix \textbf{\ref{app:vect}} for a discussion for generalisations to arbitrary Riemannian manifolds.) As before, we define a dissipation bracket on functionals $F,G$ by
\begin{align}
(F,G) & = \int \mathrm{d}^n x \mathrm{d}\Gamma \ \psi \widetilde{g}\left( \mathrm{d}\left(\frac{\delta F}{\delta \psi}\right),  \mathrm{d}\left(\frac{\delta G}{\delta \psi}\right) \right), \nonumber \\
& = \int \mathrm{d}^n x \mathrm{d}\Gamma \ \psi \left[ \delta^{ij}\frac{\partial}{\partial x^i}\left(\frac{\delta F}{\delta \psi}\right)\frac{\partial}{\partial x^j}\left(\frac{\delta G}{\delta \psi}\right) + \sum_{a=1}^N  \delta^{ij}\frac{\partial}{\partial y_{(a)}^i}\left(\frac{\delta F}{\delta \psi}\right)\frac{\partial}{\partial y_{(a)}^j}\left(\frac{\delta G}{\delta \psi}\right) \right].
\end{align}
Given a Hamiltonian functional $H$, which should now be interpreted as the free energy, the time evolution of a functional $F$ is given by
\begin{align}
\dt{F} = \{F,H\} - \frac{1}{\zeta}(F,H)
\end{align}
for some parameter ${\zeta} >0$. This can be thought of as implementing a linear mobility relation with mobility ${1}/{\zeta}$. The expression for the stress tensor in terms of the functional derivatives of $H$ is not altered by this dissipation bracket. However, the time evolution of $\psi$ will be altered, which in turn affects the closure properties of the system.

Again, consider free energies of the form (\ref{energyN}), and further specialse to energy functions $E(y)$ which are \emph{quadratic} in each internal degree of freedom, i.e.
\begin{align}\label{EN}
E(y) = \sum_{a=1}^{N} \frac{\kappa_a}{2}\delta_{ij}y_{(a)}^iy_{(a)}^j ,
\end{align}
where $\kappa_a > 0$ are parameters that describe the stiffness of each normal mode. As seen in section \textbf{\ref{subsec:closureN1}}, this energy function produces a stress that conserves angular momentum.

The evolution equation for $\psi$ will be
\begin{align}\label{psiNeomfinal}
\dt{\psi} & + \frac{\partial}{\partial x^i}\left( u^i(x) \psi \right) + \sum_{a=1}^N \frac{\partial}{\partial y^i_{(a)}}\left(  \psi \mathcal{P}_a\left(y,\frac{\partial}{\partial x}\right) \cdot u^i(x) \right) \nonumber \\
& = \frac{k_B T}{\zeta}\nabla_{x}^2 \psi + \sum_{a=1}^N \left( \frac{1}{\zeta}\frac{\partial}{\partial y^i_{(a)}}\left(\delta^{ij}\frac{\partial E}{\partial y^j_{(a)}}\psi\right) + \frac{k_B T}{\zeta}\nabla_{y_{(a)}}^2 \psi \right), \nonumber \\
& = \frac{k_B T}{\zeta}\nabla_{x}^2 \psi + \sum_{a=1}^N \left( \frac{\kappa_a}{\zeta}\frac{\partial}{\partial y^i_{(a)}}\left(y^i_{(a)}\psi\right) + \frac{k_B T}{\zeta}\nabla_{y_{(a)}}^2 \psi \right),
\end{align}
where $\nabla^2_x = \delta^{ij}{\partial^2}/{\partial x^i \partial x^j}$, and $\nabla^2_{y_{(a)}} = \delta^{ij}{\partial^2}/{\partial y^i_{(a)} \partial y^j_{(a)}}$ for $a = 1,\ldots,N$. The terms due to the dissipation bracket have been collected to the right-hand side. These new terms can be interpreted as the drift terms and the diffusive terms in a \emph{Fokker--Planck equation}. The drift terms can be attributed to the internal energy term in the free energy $H_s$, while the diffusive terms can be attributed to the entropy term in $H_s$. 

Now note that for this specific form of the energy, the time evolution of a polynomial moment $\left\langle p(y) \right\rangle$ of order $b$ depends on polynomial moments of order $b$ or lower, as can be directly verified by integrating by parts -- the drift terms do not alter the order (or gives zero), while the diffusion terms always lower the order. So the multibead-chain model with linear dissipation has a finite closure for quadratic energy.

Explicitly, we see that the stress tensor is, using (\ref{stressN1},\ref{stressN2}):
\begin{align}\label{stressN3}
\sigma^{jk} =& - (N+1) k_B T \delta^{jk} \int \mathrm{d}\Gamma \ \psi + \sum_{a=1}^N\int \mathrm{d}\Gamma \ \kappa_a y^k_{(a)} \mathcal{Q}^j_a\left(y,-\frac{\partial}{\partial x}\right) \cdot \psi, \nonumber \\ 
=& - (N+1) k_B T \delta^{jk} \int \mathrm{d}\Gamma \ \psi + \sum_{a=1}^N\int \mathrm{d}\Gamma \ \kappa_a y^j_{(a)}y^k_{(a)} \psi
+ \frac{\partial}{\partial x^l} \left(\sum_{a=1}^N\int \mathrm{d}\Gamma \ y^k_{(a)}\mathcal{R}^{jl}_a \left(y,-\frac{\partial}{\partial x}\right) \cdot \psi \right),
\end{align}
where $\mathcal{R}^{jl}_a$ is defined in terms of $\mathcal{P}_a$ and $\mathcal{Q}^j_a$ by
\begin{align}
\mathcal{P}_a\left( y, -\frac{\partial}{\partial x}\right) & = \mathcal{Q}_a^j  \left( y, -\frac{\partial}{\partial x} \right) \left( - \frac{\partial}{\partial x^j}\right) = - y_{(a)}^j \frac{\partial}{\partial x^j} + \mathcal{R}^{jl}_a\left( y, -\frac{\partial}{\partial x}\right) \frac{\partial^2}{\partial x^j \partial x^l}.
\end{align}
$\mathcal{R}^{jl}_a$ can be checked to be a polynomial from the definition (\ref{defPa}) of $\mathcal{P}_a$.  Note that $\mathcal{R}^{jl}_a$ is identically zero for $a=1$, but nonzero for $a\geq 2$. As a polynomial in $y$, $\mathcal{R}^{jl}_a$ is homogeneous of order $a$ for $a\geq 2$. So the order of the polynomial moments needed to describe the stress tensor $\sigma^{jk}$ in (\ref{stressN3}) does not exceed $2N$. This can be used to find a concrete upper bound for the number of moments required to produce a closed system.

In fact, in the evolution equation (\ref{psiNeomfinal}) for $\psi$, the hydrodynamic terms, the drift terms and the $x$-diffusion term only couples monomial moments of order $b$ to monomial moments of the same order, while the $y_{(a)}$-diffusion term couples monomial moments of order $b$ to monomial moments of order $b-2a$ (or gives $0$), so it is sufficient to consider the even order monomial moments up to order $2N$. This is consistent with the Hookean $3$-bead chain considered in section \textbf{\ref{subsec:closure2}}, where the monomial moments required are precisely those of orders $0$,$2$ and $4$.

A slightly more general sufficient condition on $E(y)$ for finite closure to be possible is that
\begin{align}\label{admissible}
\text{for all $a = 1,\ldots,N$,} \ \frac{\partial E}{\partial y^j_{(a)}} \ \text{is a polynomial in $y$ of order less than or equal to $a$.}
\end{align}
This property is evidently satisfied for quadratic energies. We will call a polynomial $p(y)$ satisfying property (\ref{admissible}) \emph{admissible}, and other polynomials \emph{inadmissible}.

To see why an admissible $E(y)$ leads to a finite closure, observe that the drift term in the evolution equation for $\psi$ due to internal relaxation is
\begin{align*}
\frac{1}{\zeta} \sum_{a=1}^N  \frac{\partial}{\partial y^i_{(a)}}\left(\delta^{ij}\frac{\partial E}{\partial y^j_{(a)}}\psi\right) .
\end{align*}
As discussed above, the hydrodynamic and diffusion terms only couple polynomial moments to other moments of equal or lower order. When $E(y)$ is admissible, the expression above suggests that a similar situation prevails -- the drift terms also only couple moments to other moments of equal or lower order. This means an \emph{order-counting} argument can be used to demonstrate the possibility of an exact, finite closure.

In more detail, since we are only considering internal energies $E(y)$ that are polynomials, the stress tensor can be described by a finite number of moments. If the maximum order of the required moments is $b_{max}$, we can show that, when $E(y)$ is admissible, collecting all monomials or order less than or equal to $b_{max}$ will result in a closed system, by the same argument used to demonstrate finite closure for quadratic energies of the form (\ref{EN}).

While $E(y)$ being admissible is nominally more general than $E(y)$ being of the Hookean-like form (\ref{EN}), we have the following highly constraining result:
\begin{prop*}
$E(y)$ is admissible (in the sense of (\ref{admissible})) if and only if \[E(y) = \sum_{a=1}^N\left( \lambda_{(a)i}y^i_{(a)} + \frac{1}{2}\kappa_{(a)ij}y^i_{(a)}y^j_{(a)} \right) + C,\] where $C, \lambda_{(a)i}, \kappa_{(a)ij}$ are constants, with  $\kappa_{(a)ij} = \kappa_{(a)ji}$.
\end{prop*}
\begin{proof}
Suppose $p(y)$ is inadmissible, and $q(y)$ is an arbitrary polynomial that is not identically zero. We claim that $p$ being inadmissible implies that $pq$ is inadmissible. To see this, suppose that
\[\frac{\partial p}{\partial y^j_{(a)}} \ \text{has order strictly greater than} \ a , \ \text{for some $a$.}\]
Then
\[\frac{\partial }{\partial y^j_{(a)}}\left(pq\right) = \frac{\partial p}{\partial y^j_{(a)}}q + p\frac{\partial q}{\partial y^j_{(a)}}\]
has order strictly greater than $a$, since $q$ is not identically zero, and ${\partial p}/{\partial y^j_{(a)}}$ has order strictly greater than $a$.

Now we can use this to show that all polynomials containing cross terms are inadmissible, since the \enquote{lowest} cross term is inadmissible. More precisely, if $p(y) = \tau_{ij}y^i_{(a)}y^j_{(b)}$ for $a \neq b$, where $\tau_{ij}$ are constants which are not identically zero, then without loss of generality we can assume $a<b$, and thus
\[ \frac{\partial}{\partial y^k_{(a)}}\left( \tau_{ij}y^i_{(a)}y^j_{(b)} \right) = \tau_{kj}y^j_{(b)} \ \text{has order $b>a$,} \]
so this $p(y)$ is inadmissible. Therefore if $E(y)$ is admissible, it cannot contain any cross terms, and hence can be written in the form
\[E\left(y_{(1)},\ldots,y_{(N)}\right) = \sum_{a=1}^N E_a\left( y_{(a)} \right)\]
for polynomials $E_a\left( y_{(a)} \right)$. Now admissibility is equivalent to
\[ \frac{\partial E(y)}{\partial y^j_{(a)}} = \frac{\partial E_a\left( y_{(a)} \right)}{\partial y^j_{(a)}} \ \text{has order less than or equal to $a$ for all $a=1,\ldots,N$}. \]
This can now be routinely checked to be equivalent to $E_a$ being linear-quadratic in $y_{(a)}$. The constants of integration for each $E_a$ can be collected together into a single constant $C$.
\end{proof}

This means that, among all the admissible (in the sense of (\ref{admissible})) choices of energy functions $E(y)$, the Hookean-like energy functions considered in (\ref{EN}) are the only ones that are rotationally invariant. Many physically plausible energy functions, such as that for a non-Hookean bead-spring pair, are not admissible. For such choices, we cannot achieve an exact, finite closure by an order-counting argument in the multibead-chain model with dissipation. In these cases, unless a new argument is found, we can only hope to find some reasonable approximations that would give \emph{approximate closure}, similar to the \emph{Peterlin approximation} for a bead-spring pair with a non-Hookean spring \cite{Bird87, Renardy00}.

\section{Conclusion}\label{sec:conclusion}

The main focus of this paper has been on the formulation and analysis of models for multibead-chain suspensions in an ideal fluid. We have modelled the suspension as a double bracket system, with of a Hamiltonian part described by a non-canonical Poisson bracket, and a dissipation bracket to account for the resistive and diffusive effects. The resulting system describes the coupling of ideal compressible hydrodynamics to the distribution function $\psi(x,y)$ of multibead-chains in configuration space, where $x$ and $y$ are the macroscopic and internal degrees of freedom respectively. This description is valid when the macroscopic lengthscales are sufficiently large, and the microscopic lengthscales of the individual chains are sufficiently small, so that the chains are advected by a Stokes flow while the Newtonian viscous stress on the fluid is negligible. If the fluid domain is a manifold $M$, then the appropriate configuration space for an $(N+1)$-bead chain is the $N^{th}$ order tangent bundle $T^{(N)}M$. The Hamiltonian part of the system, which consists of the advection of the beads in the multibead-chain by the fluid as Lagrangian markers, can be described by a Poisson bracket, which has been constructed using the machinery of the semidirect product Lie--Poisson formulation. The dissipative part of the system, consisting of the effects of internal relaxation and diffusion, can be effectively captured by a metric dissipation bracket. For suitable choices of the free energy, the non-Hamiltonian terms can be considered as a Fokker--Planck diffusion operator on the distribution function $\psi(x,y)$.

One of the main advantages for such a double bracket formulation is that, given a Hamiltonain functional (or free energy) that can be written in the form $H = H_{fluids}[\mathbf{m},\rho] + H_s[\psi]$, where $H_{fluids}$ is the usual ideal fluid Hamiltonian, the elastic body force exerted by the chains on the fluid depends on $\delta H/\delta \psi$ only and can be calculated from a direct manipulation of the Poisson bracket. In this manner we have obtained explicit expressions for the particle-contributed stress tensor for a wide range of free energies in the multibead-chain model. We found that the stress tensor is generically asymmetric, but nonetheless for reasonable choices of the internal energy $E(y)$, the antisymmetric part of the stress tensor can be written as the divergence of a $3$-index tensor, which we interpret as an angular momentum flux, as in \cite{CondiffDahler64, Cosserat09book}. This effect is absent in the bead-spring pair models, and any model that assumes a linear flow around the multibead-chain is in fact equivalent to modelling a number of bead-spring pairs with a common centre.

A major concern in using a distribution function description for the multibead-chains is that the evolution equation of $\psi(x,y)$ involves both the macroscopic ($x$) and internal ($y$) degrees of freedom. In other words, it is a partial differential equation in a large number of dimensions, which is computationally expensive to solve. However, the particle-contributed stress typically depends only on some statistical properties of the internal degrees of freedom, i.e. some $y$-integrals of the distribution function $\psi(x,y)$. Therefore it is desirable to find a finite set of these $y$-integrals, such that:
\begin{enumerate}
\item The stress tensor is completely described by these $y$-integrals.
\item These $y$-integrals form a closed system of evolution equations, given the fluid velocity field $\mathbf{u}(x,t)$ and its spatial derivatives. 
\end{enumerate}
When this is possible, we can model the multibead-chain suspension by evolving a number of macroscopic fields (depending on $x$ only), instead of having to consider both the macroscopic ($x$) and internal ($y$) degrees of freedom. We have shown that, if we choose a certain quadratic form for the internal energy $E(y)$, such a finite closure is possible for an arbitrary multibead-chain. These can be considered as analogues of the upper-convected Maxwell model for bead-spring pairs.

Within the framework of the distribution function approach, exact closure is a rare property that is only satisfied for rather restrictive choices of the internal energy. For other choices of the internal energy, we have to evolve the full distribution $\psi(x,y)$, which is much more computationally expensive than evolving $x$-dependent fields. The closure problem is clearly visible in the bead-spring pair model for non-Hookean springs, and our work heavily suggests that an entirely analogous obstacle exists for multibead-chain models.

A parallel line of development, which resolves this particular problem, would be to start with \emph{internal state variables}, which are phenomenological $x$-dependent fields that are assumed to completely characterise the internal structure of the complex fluids, inasmuch as the macroscopic fluid properties, such as the particle-contributed stress, are concerned. Mathematically, these internal state variables are typically sections of some naturally constructed fibre bundles, e.g. tensor fields, and the hydrodynamic part of the evolution of such variables would be the effect of the fluid flow as an infinitesimal coordinate transformation. For example, if the conformation tensor $C^{jk}(x)$ is chosen to be the internal state variable, the appropriate material derivative would be the familiar \emph{upper-convected derivative}, or for more general tensor fields the \emph{Lie derivative}. Analogous semidirect product Lie--Poisson structures exists for such a description, and have been extensively studied in \cite{Marsden84a,Marsden84b}, with various applications to complex fluids \cite{Grmela88, Grmela89, Beris90, Edwards90, Beris94book, Mackay19} and to magnetohydrodynamics \cite{Morrison80, Marsden84a}.

However, in this description, we have different problem to that of the distribution function approach -- the equations are easy to solve, but difficult to write down. Not only is it difficult to connect the postulated form of the free energy to a microscopic toy model of the suspended bodies, the imposition of the mobility relations has also become much more arbitrary. Unlike in the distribution function approach, there is no obvious way to convert a microscopic toy model for the suspended bodies to the relaxation and dissipation terms for the internal state variables. In other words, the appropriate form of the dissipation bracket has to be guessed. The form of dissipation bracket can be constrained by requiring material invariance and the satisfaction of thermodynamic principles, which rules out some choices as unphysical, but this does not fundamentally eliminate the arbitrariness of such a choice. Another perhaps more glaring source of arbitrariness is the choice of the internal state variables -- how many of them do we assume to be sufficient to describe the internal state phenomenologically?

The main point here is that the evolution equations for the distribution function are easy to write down but difficult to solve, which is essentially the opposite situation with the internal state variable approach. To elaborate, note that the microphysics of the suspended body is directly implemented into the evolution equation of the distribution function $\psi(x,y)$, which gives a complete (ensemble) description of the suspended bodies. The internal relaxation force is simply a gradient of the internal energy, and the diffusive Brownian force can be expediently captured in terms of the Boltzmann entropy of the distribution function $\psi(x,y)$ \cite{Grmela88}. The mobility relation is in principle determined by the internal structure of the suspended body, although we have only worked with linear approximations. The resulting equation has a clear physical interpretation as a Fokker--Planck equation, consisting of the hydrodynamic drift, internal relaxation and diffusion of the distribution function $\psi(x,y)$. Thus there is no room for arbitrariness -- once the internal microphysics of the suspended body is determined, we can immediately write down the governing equations for the fluid suspension based on this information. Moreover, when an exact closure is possible, the distribution function approach reduces to the internal state variable approach exactly, as far as the macroscopic behaviour of the fluid suspension is concerned. This process also selects the appropriate internal state variables that are necessary for a complete macroscopic description of the suspension.

Having made the case for the complementary nature of the two approaches, we note however that the distribution function approach has been largely abandoned since \cite{Grmela88,Grmela89}. We hope that this paper will revive the interest in modelling fluid suspensions with distribution functions, in particular as part of a \emph{combined approach}, where one uses the distribution function as a starting point to implement the microphysics of the suspended bodies, then proceed to make approximations to obtain effective equations in terms of internal state variables. This approach is largely unexplored beyond the bead-spring pair models, and we believe that, at the very least, this can serve as a way to inform the phenomenological approach based on internal state variables.

\section*{Ackowledgements}

The author wishes to thank Paul Dellar for bringing to our attention the connection between the asymmetric stress tensors encountered in the multibead-chain suspensions and the couple stresses in a generalised continuum system, among his numerous comments and suggestions. This work was supported by the Mathematical Institute, University of Oxford, which played no other role in the research, or in the preparation and submission of the manuscript.

\appendix
\section{Proof sketch for the homomorphism property of complete lifts to $T^{(N)}M$}\label{app:diff}

In this appendix we sketch a proof for the homomorphism property of complete lifts of vector fields $\mathbf{u}\in\mathrm{Vect}(M)$ to $\mathbf{u}^{\#}\in\mathrm{Vect}(T^{(N)}M)$ (\ref{homo1},\ref{homoN}). For more details, see \cite{Yano73book}.

Let $M$ be a manifold.  Define a \emph{diffeomorphism} $\varphi: M\rightarrow M$ to be a smooth bijective map from $M$ to $M$ with a smooth inverse. In the language of tensor calculus, a diffeomorphism $\varphi$ can be thought of as a coordinate transformation $x^i \mapsto \varphi^i(x)$. The collection of diffeomorphisms of a manifold $M$ is denoted by $\mathrm{Diff}(M)$, which can be thought of as an infinite-dimensional Lie group \cite{Khesin08,EbinMarsden70}.

Let $\varphi,\psi$ be diffeomorphisms from $M$ to itself. Consider a path $\alpha:I\rightarrow M$, where $I$ is some closed interval containing $0$ in its interior. Then the diffeomorphism $\varphi$ can act on paths $\alpha(t)$ by
\begin{align}
\varphi^{\#}: \alpha(t) \mapsto (\varphi \circ \alpha)(t).
\end{align}
Since the composition of functions is associative, this action satisfies the homomorphism property:
\begin{align}
\left(\varphi \circ \psi\right)^{\#}(\alpha) = \varphi^{\#}\left(\psi^{\#}(\alpha)\right) = \left(\varphi\circ\psi\circ\alpha\right)(t).
\end{align}
Now consider equivalence classes of paths under the equivalence relation $\stackrel{(N)}{\sim}$ defined by
\begin{align}
\alpha(t) \stackrel{(N)}{\sim} \beta(t) \Leftrightarrow \alpha(0)=\beta(0), \frac{\mathrm{d}^a}{\mathrm{d}t^a}\alpha(t)\bigg\rvert_{t=0} = \frac{\mathrm{d}^a}{\mathrm{d}t^a}\beta(t)\bigg\rvert_{t=0} \ \text{for $a=1,\cdots,N$.}
\end{align}
The equivalence relation $\stackrel{(N)}{\sim}$ identifies paths that have the same Taylor series up to order $N$ at $t=0$. Now suppose we have a path $\alpha(t)$ with coordinate expression
\begin{align}
\alpha^i(t) = x^i + \sum_{a=1}^N y^i_{(a)} \frac{t^a}{a!} + O(t^{N+1}).
\end{align}
If $\varphi^i(x)$ is the coordinate expression for $\varphi$, then comparing powers of $t$ gives coordinates of the path $\left(\varphi^{\#}(\alpha)\right)(t)$ as
\begin{align}
\left(\varphi^{\#}(\alpha)\right)^i (t) = \varphi^i(\alpha(t)) = \varphi^i(\alpha(0)) + \sum_{a=1}^N \frac{t^a}{a!} \left(\frac{\mathrm{d}^a}{\mathrm{d}t^a}\varphi^i(\alpha(t))\bigg\rvert_{t=0}\right) + O(t^{N+1}).
\end{align}
Note that
\begin{align}
\frac{\mathrm{d}^a}{\mathrm{d}t^a}\varphi^i\left( x^i + \sum_{a=1}^N y^i_{(a)} \frac{t^a}{a!} + O(t^{N+1}) \right)\bigg\rvert_{t=0}
\end{align}
does not depend on the $O(t^{N+1})$ terms when $a \leq N$.

So $\varphi^{\#}$ descends to a well-defined map on equivalence classes of paths under $\stackrel{(N)}{\sim}$, i.e. equivalence classes of paths with the same Taylor series at $t=0$ up to order $N$. We will also denote this map by $\varphi^{\#}$.

Thus we have obtained an association
\begin{align}
\varphi \in \mathrm{Diff}(M) \mapsto \varphi^{\#} \in \mathrm{Diff}(T^{(N)}M),
\end{align}
such that $\left(\varphi \circ \psi\right)^{\#} = \varphi^{\#}\circ\psi^{\#}$, i.e. it is a \emph{group homomorphism}.
By differentiating this correspondence, i.e. by writing
\begin{align}
\varphi^i(x) &= x^i + su^i(x) + O(s^2), \\
\psi^i(x) & = x^i + rv^i(x) + O(r^2),
\end{align}
and comparing the two sides of the expression
\begin{align}
\frac{\partial^2}{\partial s\partial r}\bigg\rvert_{s=0, r=0} \left(\varphi \circ\psi\circ\varphi^{-1}\right)^{\#} = \frac{\partial^2}{\partial s\partial r}\bigg\rvert_{s=0, r=0} \left(\varphi^{\#} \circ\psi^{\#}\circ\left(\varphi^{-1}\right)^{\#}\right),
\end{align}
we obtain a homomorphism of Lie algebras
\begin{align}
\mathbf{u}\in \mathrm{Vect}(M) \mapsto \mathbf{u}^{\#}\in\mathrm{Vect}\left(T^{(N)}M\right),
\end{align}
which is precisely the complete lift given in (\ref{homoN}).

\section{Vector bundle structure and metrics on $T^{(N)}M$}\label{app:vect}

This appendix addresses the geometric interpretation of some of the expressions we have encountered. We will assume some familiarity with vector bundles and Riemannian geometry. A detailed exposition can be found in, for example, \cite{Tu17book, Jost06book}.

\subsection{Vector bundle structure on $T^{(N)}M$}\label{subapp:vectTNM}

Let $T^{(N)}M$ be the $N^{th}$ order tangent bundle of the manifold $M$. As before, if $x^i$ is a coordinate system on $M$, it induces a coordinate system $(x^i,y^i_{(a)})$ on $T^{(N)}M$, where the fibre coordinates $y^i_{(a)}, \ a = 1,\cdots,N$ denote the $a^{th}$ derivatives of the equivalence class of paths attached to $x$. 

Given a coordinate transformation $x^i \mapsto {\tilde{x}}^i$, the induced coordinate transformation $(x^i, y^i_{(a)}) \mapsto (\tilde{x}^i, \tilde{y}^i_{(a)})$ is not linear for $a\geq 2$. For example,
\begin{align}
\tilde{y}^i_{(2)} = \frac{\partial \tilde{x}^i}{\partial x^j} y^j_{(2)} + \frac{\partial^2 \tilde{x}^i}{\partial x^j \partial x^k} y^j_{(1)}y^k_{(1)}.
\end{align}
(See appendix \textbf{\ref{app:diff}}.) In particular, the fibre coordinates do not transform like vectors under an arbitrary coordinate transformation.

In previous sections we have considered moments of the distribution function of the form
\begin{align}\label{appBexample}
\int \mathrm{d} \Gamma \ y^i_{(a)} \psi(x,y) \quad \text{for $a\geq 2$ (say)}, 
\end{align}
where $\mathrm{d}\Gamma = \mathrm{d}^n y_{(1)} \cdots \mathrm{d}^n y_{(N)}$ as before. Since the fibre coordinates do not transform linearly, we cannot not make invariant sense of linear operations on the fibre coordinates, such as sums and integrals.

In more geometric terms, the transition maps between two overlapping charts $(x^i, y^i_{(a)}) \mapsto (\tilde{x}^i, \tilde{y}^i_{(a)})$ of $T^{(N)}M$ are not linear in the fibre coordinates for $N\geq2$. These charts therefore do not give $T^{(N)}M$ the \emph{structure of a vector bundle}.

However, we can produce a noncanonical isomorphism from $T^{(N)}M$ to $(T\oplus\cdots\oplus T)M$ ($N$ times) with the help of a \emph{metric}, or more generally a \emph{connection} on the tangent bundle \cite{Dodson82a,Dodson82b}. We will sketch the construction as follows:

Given a metric $g$ on $M$, we can consider its \emph{Riemannian connection} $\nabla$ (also known as the \emph{Levi-Civita connection}) , which (roughly speaking) defines the infinitesimal parallel transport of vector field $Y=Y^i{\partial}/{\partial x^i}$ along vector field $X = X^i{\partial}/{\partial x^i}$ to be $\nabla_X Y$, given in coordinates by
\begin{align}
\nabla_X Y = X^j\left(\frac{\partial Y^i}{\partial x^j} + \Gamma^i_{jk} Y^k\right)\frac{\partial}{\partial x^i},
\end{align}
where $\Gamma^i_{jk}$ are the usual \emph{Christoffel symbols} of the connection, defined by
\begin{align}
\nabla_{\frac{\partial}{\partial x^j}}\frac{\partial}{\partial x^k} = \Gamma^i_{jk}\frac{\partial}{\partial x^i}.
\end{align}
The following construction will work for any connection on $TM$, not necessarily a Riemannian connection. General connections are considered in appendix \textbf{\ref{subapp:metricE}}.

Let $\alpha(t)$ be a curve on $M$, with $v(t) = \mathrm{d}\alpha(t) / \mathrm{d}t$ its velocity vector at time $t$, which is a vector field along the curve $\alpha(t)$. If $w(t)$ is another vector field along the curve $\alpha(t)$, we can define its \emph{covariant derivative}, conventionally written as $Dw / \mathrm{d}t$, along the curve $\alpha(t)$ as follows. If $W(x)$ is a vector field on $M$, such that $W(\alpha(t))=w(t)$, i.e. $W$ agrees with $w$ along the curve $\alpha(t)$, then we define
\begin{align}
\frac{Dw}{\mathrm{d}t}(t) = \nabla_{v(t)} W \left( \alpha(t) \right).
\end{align}
Applying the chain rule $({\partial W^i}/{\partial x^j})({\partial \alpha^j}/{\partial t}) = {\partial w^i}/{\partial t}$ gives the following coordinate expression for the covariant derivative:
\begin{align}
\left(\frac{Dw}{\mathrm{d}t}\right)^i = \frac{\partial w^i}{\partial t} + \Gamma^i_{jk}v^j(t)w^k(t).
\end{align}
The components $({Dw}/{\mathrm{d}t})^i$ transform like a vector under a change of coordinates, as long as the Christoffel symbols $\Gamma^i_{jk}$ are transformed appropriately.

In particular, we can apply the covariant derivative $D/\mathrm{d}t$ iteratively on the velocity vector field $v(t) = \mathrm{d}\alpha(t) / \mathrm{d}t$ along the curve $\alpha(t)$ to obtain the \emph{covariant acceleration} $Dv / \mathrm{d}t$ and higher order covariant derivatives. We will show that we can write the first $N$ coordinate derivatives of the path $\alpha$ at $t=0$ in terms of the covariant derivatives $v(0), Dv / \mathrm{d}t(0),\cdots,(D^{N-1}v/dt^{N-1})(0)$ evaluated at $t=0$, which are genuine tangent vectors at $\alpha(0)$. This will produce an isomorphism from $T^{(N)}M$ to $(T\oplus\cdots\oplus T)M$ ($N$ times), and the latter is a vector bundle by construction.

The isomorphism can be described explicitly in coordinates as follows. Let $\alpha(t)$ be a path, with coordinates
\begin{align*}
\alpha^i(t) = x^i + \sum_{a=1}^N y^i_{(a)} \frac{t^a}{a!} + O(t^{N+1}).
\end{align*}
Its velocity vector field $v(t) = \mathrm{d}\alpha(t) / \mathrm{d}t$ will have coordinates
\begin{align}
v^i(t) = y^i_{(1)} + \sum_{a=1}^{N-1} y^i_{(a)} \frac{t^a}{a!} + O(t^{N}).
\end{align}
Then by iteratively applying $D / \mathrm{d}t$ on $v(t)$, we have:
\begin{align}
v(0)^i &= y_{(1)}^i, \nonumber \\
\left(\frac{Dv}{\mathrm{d}t}(0) \right)^i &= y_{(2)}^i + \Gamma^i_{jk}y_{(1)}^jy_{(1)}^k,  \nonumber \\
\left(\frac{D^2v}{\mathrm{d}t^2}(0) \right)^i &= y_{(3)}^i + \Gamma^i_{jk}y_{(1)}^j\left( y_{(2)}^k + \Gamma^k_{lm}y_{(1)}^ly_{(1)}^m \right), \nonumber \\
 & \vdots \nonumber \\
\left(\frac{D^{N-1}v}{\mathrm{d}t^{N-1}}(0) \right)^i &= y_{(N)}^i + \Gamma^i_{jk}y_{(1)}^j\left(\frac{D^{N-2}v}{\mathrm{d}t^{N-2}}(0) \right)^k.
\end{align}
The Jacobian matrix of the coordinate transformation
\begin{align}\label{TNasvectbund}
\left( x^i,y^i_{(1)},y^i_{(2)},\cdots,y^i_{(N)} \right) \mapsto \left( x^i, v(0)^i, \left(\frac{Dv}{\mathrm{d}t}(0) \right)^i,\left(\frac{D^2v}{\mathrm{d}t^2}(0) \right)^i,\cdots, \left(\frac{D^{N-1}v}{\mathrm{d}t^{N-1}}(0) \right)^i \right)
\end{align}
is upper triangular, with blocks of the identity matrix $\delta^i_j$ on the diagonal, so in particular it is invertible and has determinant $1$.

Since each of the $(D^a v / \mathrm{d}t^a (0))^i$ transform as a tangent (or contravariant) vector, we have produced an isomorphism from $T^{(N)}M$ to $(T\oplus\cdots\oplus T)M$ ($N$ times).

Thus the expressions for $y$-moments of $\psi(x,y)$, such as those in (\ref{appBexample}), can be reinterpreted invariantly using this isomorphism. In more detail:
\begin{enumerate}
\item Perform the coordinate transformation (\ref{TNasvectbund}) on all expressions. Geometrically, this is the isomorphism $T^{(N)}M \simeq (T\oplus\cdots\oplus T)M$, so functions, vector fields, differential forms etc. on $T^{(N)}M$ can be transformed correspondingly to those on $(T\oplus\cdots\oplus T)M$.
\item In particular, transform $\psi(x,y)\mathrm{d}\Gamma$ into a density on $(T\oplus\cdots\oplus T)M$. Since the Jacobian matrix of the coordinate transformation (\ref{TNasvectbund}) has determinant $1$, we can transform the distribution function $\psi(x,y)$ like a scalar function. 
\item Compute all moments of $\psi(x,v(0),D v / \mathrm{d}t (0),\ldots,D^{N-1} v / \mathrm{d}t^{N-1} (0) )$ in the new set of coordinates. Since each of the $(D^a v / \mathrm{d}t^a (0))^i$ transforms as a vector, the new moments can now be interpreted geometrically as tensor fields (more precisely, fields of tensor densities) on $M$.
\end{enumerate}

If the manifold $M$ is the Euclidean space $\mathbb{R}^n$ with the standard metric $\delta_{ij}$, then the Christoffel symbols $\Gamma^i_{jk}$ vanish identically in Cartesian coordinates, so explicit coordinate expressions remain unmodified in this case. Nonetheless, the above procedure gives an invariant interpretation of the moments of $\psi$, which is useful when one wishes to consider non-Cartesian coordinate systems such as cylindrical polar coordinates, or manifolds that are not flat such as the surface of a sphere.

Nonetheless, the isomorphism from $T^{(N)}M$ to $(T\oplus\cdots\oplus T)M$ we have described is \emph{not canonical}, in the sense that it depends on the choice of a \emph{metric}, or more generally a \emph{connection} on $TM$. While there are complete lifts of vector fields to $T^{(N)}M$ and to $(T\oplus\cdots\oplus T)M$ respectively, neither of them depend on the choice of a connection, so their images under the noncanonical isomorphim $T^{(N)}M \simeq (T\oplus\cdots\oplus T)M$ do not coincide in general.

\subsection{Riemannian metrics on the total space of a vector bundle}\label{subapp:metricE}

Let $(M,g)$ be a Riemannian manifold, and let $\pi:E\rightarrow M$ be a vector bundle with a \emph{Riemannian structure} $h$, which is a smoothly varying, symmetric positive definite bilinear form $h_x: E_x \times E_x \rightarrow \mathbb{R}$ on each of the fibres $E_x = \pi^{-1}(x)$. We will sketch a construction of a Riemannian metric on the total space $E$ of the vector bundle, which in some sense is the closest metric to a block diagonal sum of $g$ and $h$, using a connection $\nabla$ on the vector bundle. The idea is as follows: the connection splits the tangent spaces to the total space $E$ into vertical and horizontal subspaces, and we can use $g$ on the horizontal subspace and $h$ on the vertical subspace as the inner product.

Let $q$ be the \emph{rank} of the vector bundle $E$, i.e. the dimension of the fibres $E_x = \pi^{-1}(x)$. Denote a point on the vector bundle $E$ by $(x,\xi)$, where $x$ is a point on $M$, and $\xi \in E_x$ is a vector. A \emph{(smooth) section} of the vector bundle $E$ is a smooth map $\sigma:M\rightarrow E$ such that $\pi \circ \sigma = id_M$ is the identity map on $M$. We will denote the space of all sections as $\Gamma(E)$.

To set out our notation, we will recall a few standard definitions. A \emph{connection} on a vector bundle $E$ is a map $\nabla: \text{Vect}(M) \times \Gamma(E) \rightarrow \Gamma(E)$, which formalises the idea of an infinitesimal parallel transport of a section $\sigma$ along a vector field $X$ on $M$, denoted as $\nabla_X \sigma$. It is required to satisfy the following properties:
\begin{align}
\nabla_{\left(fX+gY\right)} \sigma &= f\nabla_X \sigma + g\nabla_Y \sigma, \label{con1} \\
\nabla_X\left( a\sigma + \tau b\right) &= a\nabla_X \sigma + b\nabla_X\tau, \label{con2} \\
\nabla_X\left(f\sigma\right) &= (X\cdot f) \sigma + f \nabla_X \sigma, \label{con3}
\end{align}
for all smooth functions $f,g$ on $M$, all vector fields $X,Y$ on $M$, all real numbers $a,b$, and all sections $\sigma,\tau$ of $E$.
Condition (\ref{con1}) implies that the expression $\nabla\sigma$ can alternatively be interpreted as a section-valued $1$-form on $M$, since it is $C^\infty(M)$-linear in its first argument. Condition (\ref{con2}) implies that the connection is $\mathbb{R}$-linear in its second argument.  Condition (\ref{con3}) can be interpreted as a form of the \emph{Leibniz rule} -- in terms of section-valued $1$-forms, it is equivalent to
\begin{align}
\nabla\left(f\sigma\right) = \mathrm{d}f \cdot \sigma + f \nabla \sigma,
\end{align}
where the dot denotes the tensor product over $C^\infty(M)$, which we will suppress throughout in our notation. We will call $\nabla \sigma $ the \emph{covariant derivative} of the section $\sigma$. We can allow connections to act on \emph{local sections} $\sigma: U \rightarrow E|_U = \pi^{-1}(U)$, where $U$ is an open set in $M$, by restriction. This will be useful in finding local coordinate expressions later.

Let $U$ be a \emph{framed open set} on $M$, i.e. an open set on $M$ equipped with $q$ local sections $s_\alpha: U \rightarrow E|_U = \pi^{-1}(U), \alpha = 1,\ldots q$ such that $s_1(x),\ldots s_q(x)$ are linearly independent for all $x$ on $U$. The collection of local sections $s_\alpha$ will be called a \emph{frame}. Framed open sets always exists, since vector bundles are \emph{locally trivial}, i.e. for sufficiently small open sets $U$ of $M$, there is always an isomorphism $E|_U \simeq U \times \mathbb{R}^q$, and the choice of frame is equivalent to a choice of such an isomorphism. This is also called a \emph{local trivialisation} of the vector bundle.

Thus, a choice of frame $s_\alpha$ endows $E|_U = \pi^{-1}(U)$ with a coordinate system. In more detail, any point in $E|_U$ can be described by $\dim(M)+q$ coordinates $(x^i,\xi^\alpha)$, where
\begin{align}
(x,\xi) = (x^i,\xi^\alpha s_\alpha).
\end{align}
If $\nabla$ is a connection on the vector bundle $E$, its local expression relative to the frame $s_1(x),\cdots s_q(x)$ can be given in terms of the \emph{connection matrix} $\omega^\beta_\alpha$ as
\begin{align}
\nabla s_\alpha = \omega^\beta_\alpha s_\beta = \omega^\beta_{i \alpha}(x)\mathrm{d}x^i s_\beta,
\end{align}
where $\omega^\beta_\alpha = \omega^\beta_{i \alpha}(x)\mathrm{d}x^i$ should be interpreted as a matrix of $1$-forms on $U$. For example, if $E=TM$ is the tangent bundle, and given a coordinate neighbourhood $U$ of $M$, we choose the frame induced by coordinates $\partial/\partial x^\alpha$, for $\alpha = 1,\ldots \dim(M)$, then the connection matrix is related to the familiar Chrstoffel symbols:
\begin{align}
\nabla \frac{\partial}{\partial x^\alpha} = \Gamma^\beta_{i\alpha}\mathrm{d}x^i \frac{\partial}{\partial x^\beta}.
\end{align} 
A local section $\sigma = \sigma(x)$ can be written as $\sigma = \xi^\alpha(x)s_\alpha$ for $q$ local functions $\xi^1,\cdots,\xi^q$. Using the properties (\ref{con1},\ref{con2},\ref{con3}), we can write the covariant derivative of $\sigma = \xi^\alpha(x)s_\alpha$ as
\begin{align}
\nabla \left(\xi^\alpha s_\alpha\right) = \left( \mathrm{d}\xi^\beta + \xi^\alpha \omega^\beta_\alpha  \right) s_\beta.
\end{align}
The local section $\sigma$ is called \emph{horizontal} if $\nabla \sigma = 0$. Horizontal sections defined on open sets $U$ of $M$ do not exist a priori, but horizontal sections along (smooth) \emph{curves} on $U$ always exist, since the condition of being horizontal becomes an ordinary differential equation with smooth right-hand side -- there is always a unique and smooth solution for short times, given an initial condition. Given a curve $x(t)$ on $U$ and a corresponding point $(x(0),\xi)$ on $E_{x(0)}$, the unique horizontal local section along $x(t)$ that meets $(x(0),\xi)$ defines a curve $(x(t),\xi(t))$ on $E$, where $\xi(t) \in E_{x(t)}$. The curve $(x(t),\xi(t))$ obtained by this procedure is called the \emph{horizontal lift} of the  curve $x(t)$.
Consider the following $1$-forms on $E|_U = \pi^{-1}(U)$:
\begin{align}\label{appBdxi}
\delta \xi^\beta = \mathrm{d}\xi^\beta + \xi^\alpha \omega^\beta_{i \alpha }(x)\mathrm{d}x^i.
\end{align}
By construction, if we have a curve on $E$ which is the horizontal lift of a curve on $U$, then the tangent vectors of the curve on $E$ will be annihilated by the $1$-forms $\delta \xi^\beta$. Furthermore, it can be shown that $\delta \xi^\beta$ transforms as a vector under a change of frame: if $s_\alpha(x) \rightarrow s^\prime_{\alpha}(x) = A^\beta_\alpha(x)s_\beta(x)$, where $A^\beta_\alpha$ is a $(q\times q)$ invertible matrix of functions on $U$, then $\delta \xi^\beta = A^\beta_\alpha \delta \xi^{\prime\alpha}$.

At each point $(x,\xi)$ of a vector bundle $\pi:E\rightarrow M$, there is a canonically defined \emph{vertical subspace} $VE_{(x,\xi)}$ of the tangent space $T_{(x,\xi)}E$, defined as
\begin{align}
VE_{(x,\xi)} = \ker\left( \pi_{*}:T_{(x,\xi)}E\rightarrow T_x M \right) = \ker\left( \mathrm{d}x^i \right).
\end{align}
Equivalently, $VE_{(x,\xi)}$ is the span of the tangent vectors of curves based at $(x,\xi)$ that do not leave the fibre $E_x$. Since $E_x$ is a vector space, this gives a canonical isomorphism $VE_{(x,\xi)} \simeq E_x$, which we will suppress in our notation.

We can (noncanonically) define a \emph{horizontal subspace} $HE_{(x,\xi)}$ of $T_{(x,\xi)}E$, such that $VE_{(x,\xi)} \oplus HE_{(x,\xi)} = T_{(x,\xi)}E$, using a connection. It is defined as
\begin{align}
HE_{(x,\xi)} = \ker\left( \delta \xi^\beta \right) = \ker\left( \mathrm{d}\xi^\beta + \xi^\alpha \omega^\beta_{i \alpha }(x)\mathrm{d}x^i \right).
\end{align}
Equivalently, $HE_{(x,\xi)}$ is the span of tangent vectors of curves based at $(x,\xi)$, which arose from horizontal lifts of curves on the base manifold $M$ based at $x$. Since the $1$-forms (\ref{appBdxi}) transform appropriately under the change of frames, the definition of $HE_{(x,\xi)}$ as the kernel of the $1$-forms $\delta \xi^\beta$ is independent of the choice of frame, and only depends on the connection $\nabla$.

Now we can construct the Riemannian metric on the total space $E$ as follows. Let $X \in T_{(x,\xi)}E$. Since $VE_{(x,\xi)} \oplus HE_{(x,\xi)} = T_{(x,\xi)}E$, there are unique vectors $X_v \in VE_{(x,\xi)}$ and $X_h \in HE_{(x,\xi)}$ such that $X = X_v + X_h$. $X_v$ and $X_h$ are called the \emph{vertical component} and \emph{horizontal component} of $X$, respectively. We will use the subscripts $v$ and $h$ to denote the vertical and horizontal components for tangent vectors on $E$.

In terms of the Riemannian metric $g$ on $M$ and the Riemannian structure $h$ on the vector bundle $E$, we can define a Riemannian metric $G$ on the total space $E$ as follows:
\begin{align}
G_{(x,\xi)}(X,Y) = g_x(\pi_{*}X_h,\pi_{*}Y_h) + h_x(X_v,Y_v), \quad \text{for $X, Y \in T_{(x,\xi)}E$.}
\end{align}
We have suppressed the isomorphism $VE_{(x,\xi)} \simeq E_x$ from the notation. In coordinates, if $h_{\alpha\beta} = h(s_\alpha,s_\beta)$ are the components of the Riemannian structure, then
\begin{align}
G = g_{ij}\mathrm{d}x^i\mathrm{d}x^j + h_{\alpha\beta} \delta \xi^\alpha \delta \xi^\beta.
\end{align}

When $E=TM$ and the connection is the Riemannian connection (also known as the Levi-Civita connection) on $M$, the metric $G$ is precisely the \emph{Sasaki metric} \cite{Yano73book}. We can similarly endow $T^{(N)}M$ with a Riemannian metric, by applying the isomorphism $T^{(N)}M \simeq (T\oplus\cdots\oplus T)M$ using the Riemannian connection on $TM$ (see appendix \textbf{\ref{subapp:vectTNM}}). Since $(T\oplus\cdots\oplus T)M$ can be endowed with a Riemannian metric (using the direct sum Riemannian structure and the direct sum connection inherited from $TM$), we can pull this back to a Riemannian metric on $T^{(N)}M$. When $M = \mathbb{R}^n$ with the standard metric $\delta_{ij}$, the metric so constructed on $T^{(N)}M$ coincides with the standard metric on $\mathbb{R}^{n(N+1)}$, since the Christoffel symbols vanish identically.

\bibliography{references}{}
\bibliographystyle{abbrv}
\end{document}